\documentclass[10pt]{article}

\usepackage[utf8]{inputenc} 
\usepackage[T1]{fontenc}    
\usepackage{hyperref}       
\usepackage{url}            
\usepackage{booktabs}       
\usepackage{amsfonts}       
\usepackage{nicefrac}       
\usepackage{microtype}      

\usepackage{amsmath,amsfonts,amsthm,amssymb}
\usepackage{graphicx,float,wrapfig}
\usepackage{multirow}

\newtheorem{propo}{Proposition}[section]
\newtheorem{lemma}[propo]{Lemma}
\newtheorem{definition}[propo]{Definition}

\newtheorem{thm}{Theorem}

\usepackage{algorithm}
\usepackage{algorithmic}

\def\<{\langle}
\def\>{\rangle}
\def\E{{\mathbb E}}

\def\Pr{{\mathbb P}}

\def\Var{{\textrm{ Var }}}

\title{Estimating Mutual Information for Discrete-Continuous Mixtures\thanks{This manuscript appears in part at {\em Neural Information Processing Systems} (NIPS) 2017.}}

\author{
Weihao Gao\thanks{Coordinated Science Lab and Department of Electrical and Computer Engineering, University of Illinois at Urbana-Champaign, Urbana, IL 61801, USA; \texttt{\{wgao9,pramodv\}@illinois.edu}}, \;\;
Sreeram Kannan\thanks{Electrical Engineering Department, University of Washington, Seattle, WA 98195; \texttt{ksreeram@uw.edu}},\;\;
 Sewoong Oh\thanks{Coordinated Science Lab and Department of Industrial and Enterprise Systems Engineering, University of Illinois at Urbana-Champaign, Urbana, IL 61801, USA; \texttt{swoh@illinois.edu} }, \;\;
 Pramod Viswanath$^{\dagger}$\\
}
\date{}

\begin{document}

\maketitle

\begin{abstract}
Estimating mutual information from observed samples is a basic primitive,
useful in several machine learning tasks including correlation mining, information bottleneck clustering, learning a Chow-Liu tree, and conditional independence testing in (causal) graphical models.
While mutual information is a well-defined quantity in general probability spaces,
existing estimators can only handle
 two special cases of purely discrete or
purely continuous pairs of random variables.
The main challenge is that these methods first
estimate the (differential) entropies of
$X$, $Y$ and the pair $(X,Y)$ and
 add them up with appropriate signs to get an estimate of the mutual information.
These 3H-estimators cannot be applied in general mixture spaces,
where entropy is not well-defined.
In this paper, we design a novel estimator for  mutual information of discrete-continuous mixtures.
We prove that the proposed estimator is consistent.
We provide numerical experiments suggesting superiority of the proposed estimator
compared to other heuristics of
adding small continuous noise to all the samples and applying standard estimators tailored for purely continuous variables,
and quantizing the samples and applying standard estimators tailored for purely discrete variables.
 This significantly widens the
 applicability of mutual information estimation in  real-world applications,
 where some variables are discrete, some continuous, and others are a mixture between continuous and discrete components.

\end{abstract}

\section{Introduction}




A fundamental quantity of interest in machine learning is mutual information (MI), which characterizes the shared information between a pair of random variables $(X,Y)$.
MI obeys several appealing properties including the data-processing inequality, invariance under one-to-one transformations and the chain rule \cite{cover1991information},
which led to a wide use in  canonical tasks such as classification~\cite{peng2005feature}, clustering~\cite{muller2012information,ver2014maximally,chan2015multivariate} and feature selection~\cite{battiti1994using,fleuret2004fast}.
Mutual information also emerges as the ``correct''  quantity in several graphical model inference problems (e.g.,  the Chow-Liu tree \cite{chow1968approximating} and conditional independence testing \cite{bishop2006pattern}).
MI is also pervasively used in many data science application domains, such as sociology~\cite{reshef2011detecting}, computational biology~\cite{Krishnaswamy14}, and computational neuroscience~\cite{rieke1999spikes}.

An important problem in any of these applications is to estimate mutual information effectively from samples. While mutual information has been the {\em de facto} measure of information in several applications for decades, the estimation of mutual information from samples remains an active research problem. Recently, there has been a resurgence of interest in entropy and mutual information estimators, on both the theoretical as well as practical fronts~\cite{sricharan2013ensemble,moon2017ensemble,singh2016finite,singh2017nonparanormal,jiao2014maximum,han2015adaptive,gao2014efficient,gao2015estimating,gao2016demystifying,gao2016breaking}.

%
%
%

The previous estimators focus on either of two cases -- the data is either purely discrete or purely continuous. In these special cases, the mutual information can be calculated based on the three (differential) entropies of $X$, $Y$ and $(X,Y)$. We term estimators based on this principle as $3H$-estimators (since they estimate three entropy terms), and a majority of previous estimators fall under this category  \cite{han2015adaptive,gao2016breaking,sricharan2013ensemble}.

In practical downstream applications, we often have to deal with a {\em mixture of continuous and discrete} random variables. Random variables can be mixed in several ways. First, one random variable can be discrete whereas the other is continuous. For example, we want to measure the strength of relationship between children's age and height, here age $X$ is discrete and height $Y$ is continuous. Secondly, a single scalar random variable itself can be a mixture of discrete and continuous components. For example, consider $X$ taking a zero-inflated-Gaussian distribution, which takes value $0$ with probability $p$ and is a Gaussian distribution with mean $\mu$ with probability $1-p$. This distribution has both a discrete component as well as a component with density. Finally, $X$ and / or $Y$ can be high dimensional vector, each of whose components may be discrete, continuous or mixed.

In all of the aforementioned {\em mixed} cases, mutual information is  well-defined through the Radon-Nikodym derivative (see Section~\ref{sec:def}) but cannot be expressed as a function of the entropies or differential entropies of the random variables. Crucially, entropy is not well defined when a single scalar random variable comprises of both discrete and continuous components, in which case, $3H$ estimators (the vast majority of prior art) cannot be directly employed.  In this paper, we address this challenge by proposing an estimator that can handle all these cases of mixture distributions. The estimator directly estimates the Radon-Nikodym derivative using the $k$-nearest neighbor distances from the samples; we  prove $\ell_2$ consistency of the estimator and demonstrate its excellent practical performance through a variety of experiments on both synthetic and real dataset. Most relevantly, it strongly outperforms natural baselines of discretizing the mixed random variables (by quantization) or making it continuous by adding a small Gaussian noise.

The rest of the paper is organized as follows. In Section~\ref{sec:def}, we review the general definition of mutual information for Radon-Nikodym derivative. In Section~\ref{sec:estimator}, we propose our estimator of mutual information for mixed random variables.
 In Section~\ref{sec:proof}, we prove that our estimator is $\ell_2$ consistent under certain  technical  assumptions and verify that the assumptions are satisfied for most practical cases.
 Section~\ref{sec:simulation} contains the results of our  detailed synthetic and real-world  experiments testing the efficacy of the proposed estimator.

\section{Problem Formation}
\label{sec:def}
In this section, we define mutual information for general distributions as follows (e.g.,   \cite{polyanskiy2015strong}).

\begin{definition}
Let $P_{XY}$ be a probability measure on the space $\mathcal{X} \times \mathcal{Y}$, where $\mathcal{X}$ and $\mathcal{Y}$ are both Euclidean spaces. For any measurable set $A \subseteq \mathcal{X}$ and $B \subseteq \mathcal{Y}$, define $P_X(A) = P_{XY}(A \times \mathcal{Y})$ and $P_Y(B) = P_{XY}(\mathcal{X} \times B)$. Let $P_X P_Y$ be the product measure $P_X \times P_Y$. If $P_{XY}$ is absolutely continuous w.r.t. $P_XP_Y$, then the mutual information $I(X;Y)$ of $P_{XY}$ is defined as
\begin{eqnarray}
I(X;Y) &\equiv& \int_{\mathcal{X} \times \mathcal{Y}} \log \frac{dP_{XY}}{dP_XP_Y} dP_{XY}, \label{def:mi}
\end{eqnarray}
where $\frac{dP_{XY}}{dP_XP_Y}$ is the Radon-Nikodym derivative.
\end{definition}

%
%
Notice that this general definition includes the following cases of mixtures: (1) $X$ is discrete and $Y$ is continuous (or vice versa); (2) $X$ or $Y$ has many components each, where some components are discrete and some are continuous; (3) $X$ or $Y$ or their joint distribution is a mixture of continuous and discrete distributions.


\section{Estimators of Mutual Information}
\label{sec:estimator}

\subsection{Review of Previous Works}

 The estimation problem is quite different depending on whether the underlying distribution is discrete, continuous or mixed. As pointed out earlier, most existing estimators for mutual information are based on the $3H$ principle: they estimate the three entropy terms first. This $3H$ principle can be applied only in the purely discrete or purely continuous case.

{\em Discrete data}: For entropy estimation of a discrete variable $X$, the straightforward  approach  to plug-in the estimated probabilities $\hat{p}_X(x)$ into the formula for entropy has been shown to be suboptimal \cite{paninski2003estimation,acharyamaximum}. Novel entropy estimators with sub-linear sample complexity have been proposed \cite{valiant2011estimating,wu2016minimax,han2015adaptive,jiao2015minimax,han2015minimax,jiao2014non}. MI estimation can then be performed using the $3H$ principle, and such an approach is shown to be worst-case optimal for mutual-information estimation  \cite{han2015adaptive}.

{\em Continuous data}: There are several estimators for differential entropy of continuous random variables, which have been exploited in a $3H$ principle to calculate the mutual information \cite{beirlant1997nonparametric}. One family of   entropy estimators are based on kernel density estimators \cite{paninski2008undersmoothed} followed by re-substitution estimation. An alternate family of entropy estimators is based on $k$-Nearest Neighbor ($k$-NN) estimates, beginning with the pioneering work of Kozachenko and Leonenko \cite{KL87} (the so-called KL estimator). Recent progress involves an inspired mixture of an  ensemble of kernel and $k$-NN estimators  \cite{sricharan2013ensemble,berrett2016efficient}. Exponential concentration bounds under certain conditions are in \cite{singh2014exponential}. 

{\em Mixed Random Variables}: Since the entropies themselves may not be well defined for mixed random variables, there is no direct way to apply the $3H$ principle. However, once the data is quantized, this principle can be applied in the discrete domain. That mutual information in arbitrary measure spaces can indeed be computed as a maximum over quantization is a  classical result  \cite{gelfand1959calculation,perez1959information,pinsker1960information}. However, the choice of quantization is complicated and while some quantization schemes are known to be consistent when there is a joint density~\cite{darbellay1999estimation}, the mixed case is  complex. Estimator of the average of Radon-Nikodym derivative $dP/dQ$ has been studied in~\cite{wang2005divergence,wang2009divergence}. Very recent work generalizing the ensemble entropy estimator when some components are discrete and others continuous is in \cite{moon2017ensemble}.


{\em Beyond $3H$ estimation}: In an inspired work~\cite{Kra04}  proposed a {\em direct} method for estimating mutual information (KSG estimator) when the variables have a joint density. The estimator starts with the $3H$ estimator based on differential entropy estimates based on the $k$-NN estimates, and employs a heuristic to couple the estimates in order to improve the estimator. While the original paper did not contain any theoretical proof, even of consistency, its excellent practical performance has encouraged widespread adoption. Recent work \cite{gao2016demystifying} has established the consistency of this estimator along with its  convergence rate. Further, recent works~\cite{gao2015estimating,gao2016breaking} involving a combination of kernel density estimators and $k$-NN methods have been proposed to further improve the KSG estimator. \cite{Ros14} extends the KSG estimator to the case when one variable is  discrete and another is scalar continuous.

None of these works consider a case even if one of the components has a mixture of continuous and discrete distribution, let alone for general probability distributions.
There are two generic options: (1) one can add small independent noise on each sample to break the multiple samples and apply a continuous valued MI estimator (like KSG), or (2) quantize and apply discrete MI estimators but the performance for high-dimensional case is poor. These form baselines to  compare against in our detailed simulations.

\subsection{Mixed Regime}

We first examine the behavior of other estimators in the mixed regime, before proceeding to develop our estimator.  Let us consider the case when $X$ is discrete (but real valued) and $Y$ possesses a density. In this case, we will examine the consequence of using the $3H$ principle, with differential entropy estimated by the $k$-nearest neighbors.  To do this, fix a parameter $k$, that determines the number of neighbors and let $\rho_{i,x}$, $\rho_{i,y}$ and $\rho_{i,xy}$ denote the distance of the $k$-nearest neighbor of $X_i$, $Y_i$ and $(X_i,Y_i)$, respectively. Then


\begingroup\makeatletter\def\f@size{8}\check@mathfonts
$$\widehat{I}^{(N)}_{\rm 3H}(X;Y) = \left(\,\frac{1}{N} \sum_{i=1}^N   \log \frac{N c_{x} \rho_{i,x}^d }{k} +  a(k)  \,\right)+ \left(\,\frac{1}{N} \sum_{i=1}^N  \log \frac{Nc_{y} \rho_{i,y}^d }{k} +  a(k) \,\right) - \left(\, \frac{1}{N} \sum_{i=1}^N   \log \frac{N c_{xy} \rho_{i,xy}^d}{k} +  a(k) \,\right)
$$
\endgroup
where $\psi(\cdot)$ is the digamma function and $a(\cdot) = \log(\cdot)-\psi(\cdot)$.  In the case that $X$ is discrete and $Y$ has a density,
$ I_{\rm 3H}(X;Y)  =   -\infty +  a  -  b =  -\infty $,  which is clearly wrong.

The basic idea of the KSG estimator is to ensure that the $\rho$ is the same for both $x$, $y$ and $(x,y)$ and the difference is instead in the number of nearest neighbors. Let $n_{x,i}$ be the number of samples of $X_i$'s within distance $\rho_{i,xy}$ and $n_{y,i}$ be the number of samples of $Y_i$'s within  distance $\rho_{i,xy}$. Then the KSG estimator is given by
 $   \widehat{I}_{KSG}^{(N)} \equiv \frac{1}{N} \sum_{i=1}^N \left(\, \psi(k) + \log(N) - \log(n_{x,i}+1) - \log(n_{y,i}+1) \,\right)$
where $\psi(\cdot)$ is the digamma function.

In the case of $X$ being discrete and $Y$ being continuous, it turns out that the KSG estimator does {\em not} blow up (unlike the $3H$ estimator), since the distances do not go to zero. However, in the mixed case, the estimator has a non-trivial bias due to discrete points and is no longer consistent.


\subsection{Proposed Estimator}
We propose the following estimator for general probability distributions, inspired by the KSG estimator. The intuition is as follows. First notice that MI is the average of the logarithm of Radon-Nikodym derivative, so we compute the Radon-Nikodym derivative for each sample $i$ and take the empirical average. The re-substitution estimator for MI is then given as follows:
$\widehat{I}(X;Y) \equiv \frac{1}{n} \sum_{i=1}^n \log \left( \frac{dP_{XY}}{dP_XP_Y} \right)_{(x_i,y_i)}.$ 
 The basic idea behind our estimate of the Radon-Nikodym derivative at each sample point is as follows:

\begin{itemize}
\item When the point is discrete (which can be detected by checking if the $k$-nearest neighbor distance of data $i$ is zero), then we can assert that data $i$ is in a discrete component, and we can use plug-in estimator for Radon-Nikodym derivative.
\item If the point is such that there is a joint density (locally), the KSG estimator suggests a natural idea: fix the radius and estimate the Radon-Nikodym derivative by $\left( \psi(k) + \log(N) - \log(n_{x,i}+1) - \log(n_{y,i}+1) \right)$.
\item If $k$-nearest neighbor distance is not zero, then it may be either purely continuous or mixed. But we show below that the method for purely continuous is also applicable for mixed.
\end{itemize}

Precisely, let $n_{x,i}$ be the number of samples of $X_i$'s within distance $\rho_{i,xy}$ and $n_{y,i}$ be the number of samples of $Y_i$'s with in $\rho_{i,xy}$. Denote $\tilde{k}_i$ by the number of tuples $(X_i, Y_i)$ within distance $\rho_{i,xy}$. If the $k$-NN distance is zero, which means that the sample $(X_i,Y_i)$ is a discrete point of the probability measure, we set $k$ to $\tilde{k}_i$, which is the number of samples that have the same value as $(X_i,Y_i)$. Otherwise we just keep $\tilde{k}_i$ as $k$. Our proposed estimator is described in detail in Algorithm~\ref{algo1}.

\begin{algorithm}[htbp]
\caption{Mixed Random Variable Mutual Information Estimator}
    \begin{algorithmic}
        \STATE \textbf{Input:} $\{X_i, Y_i\}_{i=1}^N$, where $X_i \in \mathcal{X}$ and $Y_i \in \mathcal{Y}$;
        \STATE \textbf{Parameter:} $k \in \mathbb{Z}^+$;
        \FOR {$i = 1$ to $N$}
                \STATE $\rho_{i,xy} := $ the $k$ smallest distance among $\left[ d_{i,j} := \max\{\|X_j-X_i\|, \|Y_j-Y_i\|\}, j \neq i \right]$;
                \IF {$\rho_{i,xy} = 0$}
                    \STATE $\tilde{k}_i := $ number of samples such that $d_{i,j} = 0$;
                \ELSE
                    \STATE $\tilde{k}_i := k$;
                \ENDIF
                \STATE $n_{x,i} := $ number of samples such that $\|X_j - X_i\| \leq \rho_{i,xy}$;
                \STATE $n_{y,i} := $ number of samples such that $\|Y_j - Y_i\| \leq \rho_{i,xy}$;
                \STATE $\xi_i := \psi(\tilde{k}_i) + \log N - \log (n_{x,i}+1) - \log (n_{y,i}+1)$;
        \ENDFOR
        \STATE \textbf{Output:} $\widehat{I}^{(N)}(X;Y) := \frac{1}{N}\sum_{i=1}^N \xi_i$.
    \end{algorithmic}
    \label{algo1}
\end{algorithm}

We note that our estimator recovers previous ideas in several canonical settings.  If the underlying distribution is purely discrete, the $k$-nearest neighbor distance $\rho_{i,xy}$ equals to 0 with high probability, then our estimator recovers the plug-in estimator. If the underlying distribution is purely continuous, then there are no multiple overlapping samples, so $\tilde{k}_i$ equals to $k$, our estimator recovers the KSG estimator. If $X$ is discrete and $Y$ is single-dimensional continuous and $P_X(x) > 0$ for all $x$, for sufficiently large dataset, the $k$-nearest neighbors of sample $(x_i,y_i)$ will be located on the same $x_i$ with high probability. Therefore, our estimator recovers the discrete vs continuous estimator in \cite{Ros14}.

\section{Proof of Consistency}
\label{sec:proof}
We show that under certain technical conditions on the joint probability measure, the proposed estimator is consistent. We begin with the following definitions. Let $f(x,y) = dP_{XY}/dP_XP_Y$ denote the Radon-Nikodym derivative and define  
\begin{eqnarray}
P_{XY}(x,y,r) &\equiv& P_{XY}\left(\, \{(a,b) \in \mathcal{X} \times \mathcal{Y}: \|a-x\| \leq r, \|b-y\| \leq r\} \,\right), \\
P_X(x,r) &\equiv& P_X\left(\, \{a \in \mathcal{X}: \|a-x\| \leq r\} \,\right), \\
P_Y(y,r) &\equiv& P_Y\left(\, \{b \in \mathcal{Y}: \|b-y\| \leq r\} \,\right).
\end{eqnarray}

\begin{thm}
\label{thm:bias}
Suppose that
\begin{enumerate}
    \item $k$ is chosen to be a function of $N$ such that $k_N \to \infty$ and $k_N \log N/N \to 0$ as $N \to \infty$.
    \item The set of discrete points $\{(x,y): P_{XY}(x,y,0) > 0\}$ is finite.
    \item $\int_{\mathcal{X} \times \mathcal{Y}} \big|\, \log \frac{dP_{XY}}{dP_XP_Y} \,\big|\, dP_{XY} < +\infty$.
\end{enumerate}
Then we have
$\lim_{N \to \infty} \E \left[\, \widehat{I}^{(N)}(X;Y) \,\right] = I(X;Y) \;.$
\end{thm}

Notice that the assumptions are satisfied whenever (1) the distribution is (finitely) discrete; (2) the distribution is continuous; (3) some dimensions are (countably) discrete and some dimensions are continuous; (4) a (finite) mixture of the previous cases. Most real world data can be covered by these cases.
A sketch of the proof is below with the full proof in the supplementary material.

\begin{proof}
(Sketch) We start with an explicit form of the Radon-Nikodym derivative  $dP_{XY}/(dP_XP_Y)$.

\begin{lemma}
\label{lem:rnd}
For almost every $(x,y) \in \mathcal{X} \times \mathcal{Y}$, we have
\begin{eqnarray}
\frac{dP_{XY}}{dP_XP_Y}(x,y) = f(x,y) = \lim_{r \to 0} \frac{P_{XY}(x,y,r)}{P_X(x,r)P_Y(y,r)}.
\end{eqnarray}
\end{lemma}

Notice that $\widehat{I}_N(X;Y) = (1/N) \sum_{i=1}^N \xi_i$, where all $\xi_i$ are identically distributed. Therefore, $\E[\widehat{I}^{(N)}(X;Y)] = \E[\xi_1]$. Therefore, the bias can be written as:
\begin{eqnarray}
\Big|\, \E[\widehat{I}^{(N)}(X;Y)] - I(X;Y) \,\Big| &=& \Big|\, \E_{XY} \left[ \E \left[ \xi_1 | X,Y \right] \right] - \int \log f(X,Y) P_{XY} \,\Big| \,\notag\\
&\leq& \int \Big|\, \E \left[ \xi_1| X,Y \right] - \log f(X,Y) \,\Big| \, dP_{XY} \,.
\end{eqnarray}

Now we upper bound $\Big|\, \E \left[\, \xi_1 | X,Y \,\right] - \log f(X,Y) \,\Big|$ for every $(x,y) \in \mathcal{X} \times \mathcal{Y}$ such that Lemma~\ref{lem:rnd} is satisfied. Note that the probability of having $(x,y)$ not satisfying Lemma~\ref{lem:rnd} is zero, so we ignore this case. We then divide the domain into three parts as $\mathcal{X} \times \mathcal{Y} = \Omega_1 \bigcup \Omega_2 \bigcup \Omega_3$ where
\begin{itemize}
    \item $\Omega_1 = \{(x,y): f(x,y) = 0\} \,;$
    \item $\Omega_2 = \{(x,y): f(x,y) > 0, P_{XY}(x,y,0) > 0\}\,;$
    \item $\Omega_3 = \{(x,y): f(x,y) > 0, P_{XY}(x,y,0) = 0\} \;.$
\end{itemize}
We  show that $\lim_{N \to \infty} \int_{\Omega_i} \Big|\, \E \left[ \xi_1 | X,Y\right] - \log f(X,Y) \,\Big| \,dP_{XY} = 0$ for each $i \in \{1,2,3\}$ separately.

\begin{itemize}
    \item For $(x,y) \in \Omega_1$, we will show that $\Omega_1$ has zero probability with respect to $P_{XY}$, i.e. $P_{XY}(\Omega_1) = 0$. Hence, $\int_{\Omega_1} \Big|\, \E \left[ \xi_1 | X,Y\right] - \log f(X,Y) \,\Big| \,dP_{XY} = 0$.
    \item For $(x,y) \in \Omega_2$, $f(x,y)$ equals to $P_{XY}(x,y,0)/P_X(x,0)P_Y(y,0)$, so it can be viewed as a discrete part. We will first show that the $k$-nearest neighbor distance $\rho_{k,1} = 0$ with high probability. Then we will use the the number of samples on $(x,y)$ as $\tilde{k}_i$, and we will show that the mean of estimate $\xi_1$ is closed to $\log f(x,y)$.
    \item For $(x,y) \in \Omega_3$, it can be viewed as a continuous part. We use the similar proof technique as~\cite{Kra04} to prove that the mean of estimate $\xi_1$ is closed to $\log f(x,y)$.
\end{itemize}
\end{proof}

The following theorem bounds the variance of the proposed estimator.
\begin{thm}
\label{thm:var}
Assume in addition that
\begin{enumerate}
    \item[6.] $(k_N \log N)^2 /N \to 0$ as $N \to \infty$.
\end{enumerate}
Then we have
\begin{eqnarray}
\lim_{N \to \infty} \Var \left[\, \widehat{I}^{(N)}(X;Y) \,\right] = 0 \;.
\end{eqnarray}
\end{thm}

\begin{proof}
(Sketch) We use the Efron-Stein inequality to bound the variance of the estimator. For simplicity, let $\widehat{I}^{(N)}(Z)$ be the estimate based on original samples $\{Z_1, Z_2, \dots, Z_N\}$, where $Z_i = (X_i, Y_i)$, and $\widehat{I}^{(N)}(Z_{\setminus j})$ is the estimate from $\{Z_1, \dots, Z_{j-1}, Z_{j+1}, \dots, Z_N\}$. Then a certain version of Efron-Stein inequality states that:
$\Var \left[\, \widehat{I}^{(N)}(Z) \,\right] \leq 2 \sum_{j=1}^N \left(\, \sup_{Z_1, \dots, Z_N} \Big|\, \widehat{I}^{(N)}(Z) - \widehat{I}^{(N)}(Z_{\setminus j}) \,\Big| \,\right)^2 \;.$
Now recall that \begin{eqnarray}
\widehat{I}^{(N)}(Z) = \frac{1}{N} \sum_{i=1}^N \xi_i(Z) = \frac{1}{N} \sum_{i=1}^N \left(\, \psi(\tilde{k}_i) + \log N - \log (n_{x,i}+1) - \log (n_{y,i}+1) \,\right) \;,
\end{eqnarray}
Therefore, we have
\begin{eqnarray}
\sup_{Z_1, \dots, Z_N} \Big|\, \widehat{I}^{(N)}(Z) - \widehat{I}^{(N)}(Z_{\setminus j}) \,\Big| \leq \frac{1}{N} \sup_{Z_1, \dots, Z_N} \sum_{i=1}^N  \Big|\, \xi_i(Z) - \xi_i(Z_{\setminus j}) \,\Big| \;.
\end{eqnarray}
To upper bound the difference $|\, \xi_i(Z) - \xi_i(Z_{\setminus j}) \,|$ created by eliminating sample $Z_j$ for different $i$ 's we  consider three different cases: (1) $i=j$; (2) $\rho_{k,i} = 0$; (3) $\rho_{k,i} > 0$, and conclude that $\sum_{i=1}^N  |\, \xi_i(Z) - \xi_i(Z_{\setminus j}) \,| \leq O(k \log N)$ for all $Z_i$'s. The detail of the case study is in Section.~\ref{sec:var} in the supplementary material. Plug it into Efron-Stein inequality, we obtain:
\begin{eqnarray}
&&\Var \left[\, \widehat{I}^{(N)}(Z) \,\right] \leq 2 \sum_{j=1}^N \left(\, \sup_{Z_1, \dots, Z_N} \Big|\, \widehat{I}^{(N)}(Z) - \widehat{I}^{(N)}(Z_{\setminus j}) \,\Big| \,\right)^2 \,\notag\\
&\leq& 2 \sum_{j=1}^N \left(\, \frac{1}{N} \sup_{Z_1, \dots, Z_N} \sum_{i=1}^N  \Big|\, \xi_i(Z) - \xi_i(Z_{\setminus j}) \,\Big| \,\right)^2 = O((k \log N)^2 /N)\;.
\end{eqnarray}
By Assumption 6, we have $\lim_{N \to \infty} \Var \left[\, \widehat{I}^{(N)}(Z) \,\right] = 0$.
\end{proof}

Combining Theorem~\ref{thm:bias} and Theorem~\ref{thm:var}, we have the $\ell_2$ consistency of $\widehat{I}^{(N)}(X;Y)$.


\section{Simulations}
\label{sec:simulation}
We evaluate the performance of our  estimator in a variety of (synthetic and real-world) experiments.

\begin{figure}[h]
	\begin{center}
	\includegraphics[width=.4\textwidth]{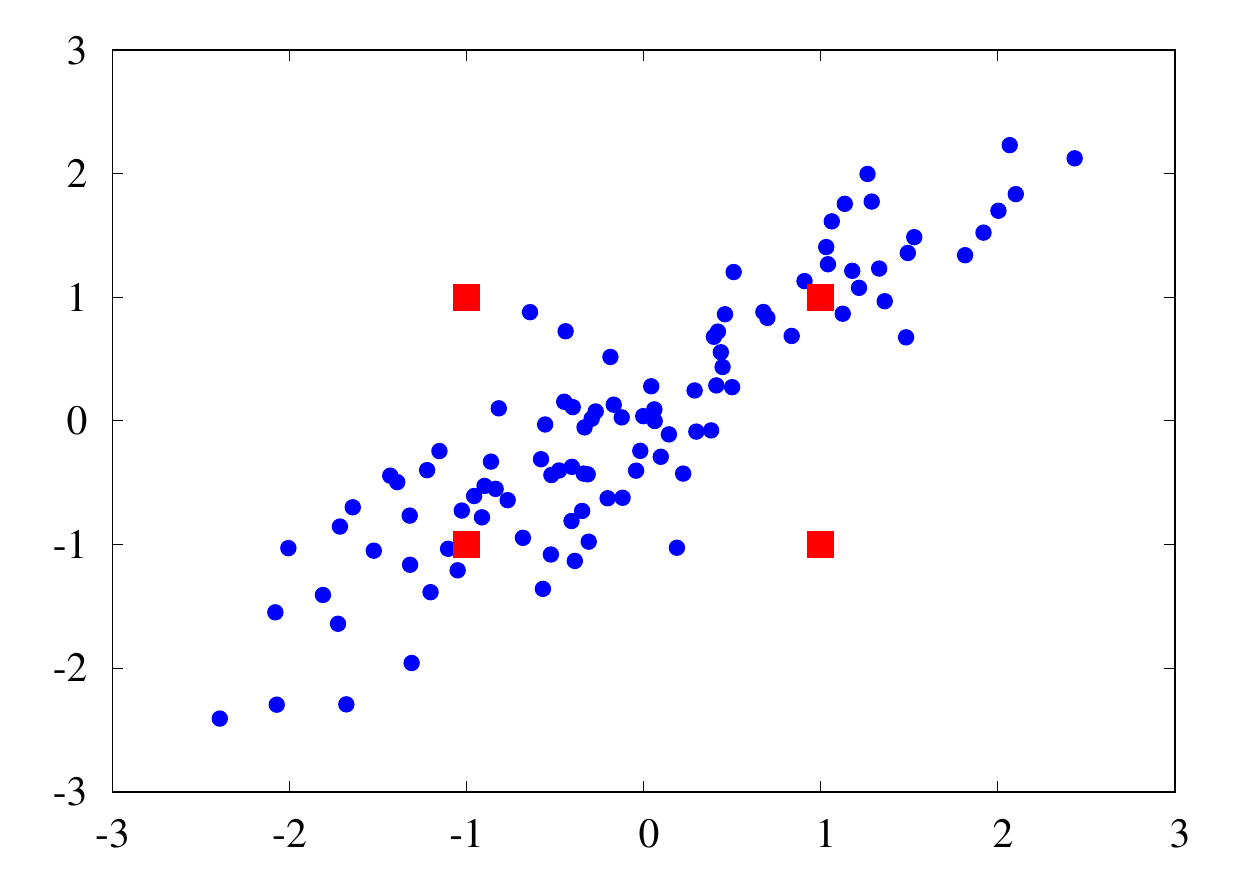}
    \put(-100,-8){$X$}
    \put(-195,65){$Y$}
	\hspace{0.5 cm}
	\includegraphics[width=.4\textwidth]{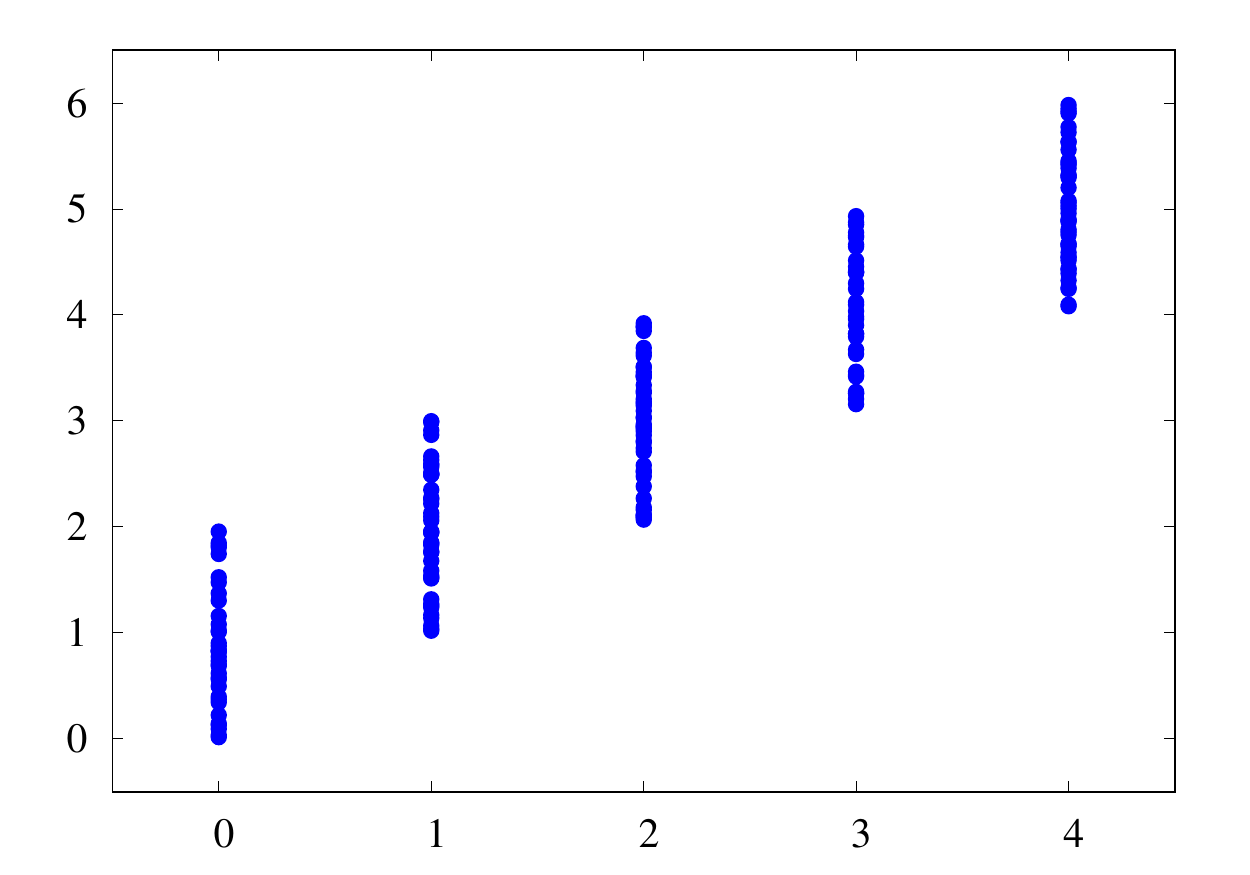}
    \put(-100,-8){$X$}
    \put(-195,65){$Y$}
	\end{center}
	\caption{Left: An example of samples from a mixture of continuous (blue) and discrete (red) distributions, where  red points denote multiple samples. Right: An example of samples from a discrete $X$ and a continuous $Y$.}
	\label{fig:example}
\end{figure}

{\bf Experiment I}. $(X,Y)$ is a mixture of one continuous distribution  and one discrete distribution.
The continuous distribution is jointly Gaussian with zero mean and covariance $\Sigma = \begin{pmatrix} 1 & 0.9 \\ 0.9 & 1\end{pmatrix}$, and the discrete distribution is $P(X=1, Y=1) = P(X=-1, Y=-1) = 0.45$ and $P(X=1, Y=-1) = P(X=-1, 1) = 0.05$. These two distributions are mixed with equal probability. The scatter plot of a set of samples from this distribution is shown in the left panel of Figure.~\ref{fig:example}, where the red squares denote multiple samples from the discrete distribution. For all  synthetic experiments, we compare our proposed estimator with a (fixed) partitioning estimator, an adaptive partitioning estimator~\cite{darbellay1999estimation} implemented by~\cite{szabo14information}, the KSG estimator~\cite{Kra04} and noisy KSG estimator (by adding Gaussian noise $N(0,\sigma^2 I)$ on each sample to transform all mixed distributions into continuous one). We plot the mean squared error versus number of samples in  Figure~\ref{fig:mse}. The mean squared error is averaged over 250 independent trials.

The KSG estimator is entirely misled by the discrete samples as expected.
The noisy KSG estimator performs better
but the added noise causes the estimate to degrade. In this experiment, the estimate is less sensitive to the noise added and the line is indistinguishable with the line for KSG.
The partitioning and
adaptive partitioning method quantizes all samples, resulting in an extra quantization error.
Note that only the proposed estimator has error decreasing with the sample size.


{\bf Experiment II}. $X$ is a discrete random variable and $Y$ is a continuous random variable.
 $X$ is uniformly distributed over integers $\{0, 1, \dots, m-1\}$ and $Y$ is uniformly distributed over the range $[X, X+2]$ for a given $X$. The ground truth $I(X;Y) = \log(m) - (m-1)\log(2)/m$. We choose $m=5$ and  a scatter plot of a set of samples is in the right panel of Figure.~\ref{fig:example}. Notice that in this case (and the following experiments) our proposed estimator degenerates to KSG if the hyper parameter $k$ is chosen the same, hence KSG is not plotted. In this experiment our proposed estimator outperforms other methods.


{\bf Experiment III.} Higher dimensional mixture.
 Let $(X_1, Y_1)$ and $(Y_2, X_2)$ have the same joint distribution as in experiment II and independent of each other.  We evaluate the mutual information between $X = (X_1, X_2)$ and $Y = (Y_1, Y_2)$. Then ground truth $I(X;Y) = 2(\log(m) - (m-1)\log(2)/m)$.  We also consider $X = (X_1, X_2, X_3)$ and $Y = (Y_1, Y_2, Y_3)$ where $(X_3, Y_3)$ have the same joint distribution as in experiment II and independent of $(X_1,Y_1), (X_2,Y_2)$. The ground truth  $I(X;Y) = 3(\log(m) - (m-1)\log(2)/m)$. The adaptive partitioning algorithm works only for one-dimensional $X$ and  $Y$ and is not compared here.

We can see that the performance of partitioning estimator is very bad because the number of partitions grows exponentially with dimension. Proposed algorithm suffers less from the curse of dimensionality. For the right figure, noisy KSG method has smaller error, but we point out that it is unstable with respect to the noise level added: as the noise level is varied from $\sigma=0.5$ to $\sigma=0.7$ and the performance varies significantly (far from convergence).

{\bf Experiment IV.} Zero-inflated Poissonization.
Here $X \sim {\rm Exp}(1)$ is a standard exponential random variable, and $Y$ is zero-inflated Poissonization of $X$, i.e., $Y = 0$ with probability $p$ and $Y \sim {\rm Poisson}(x)$ given $X = x$ with probability $1-p$. { Here the ground truth is $I(X;Y) =  (1-p)(2\log2 - \gamma - \sum_{k=1}^{\infty} \log k \cdot 2^{-k}) \approx (1-p)0.3012$, where $\gamma$ is Euler-Mascheroni constant. We repeat the experiment for no zero-inflation ($p=0$) and for $p = 15\%$. We find that the proposed estimator is comparable to adaptive partitioning for no zero-inflation and outperforms others for 15\% zero-inflation.

\begin{figure}
    \begin{center}
        \includegraphics[width=.45\textwidth]{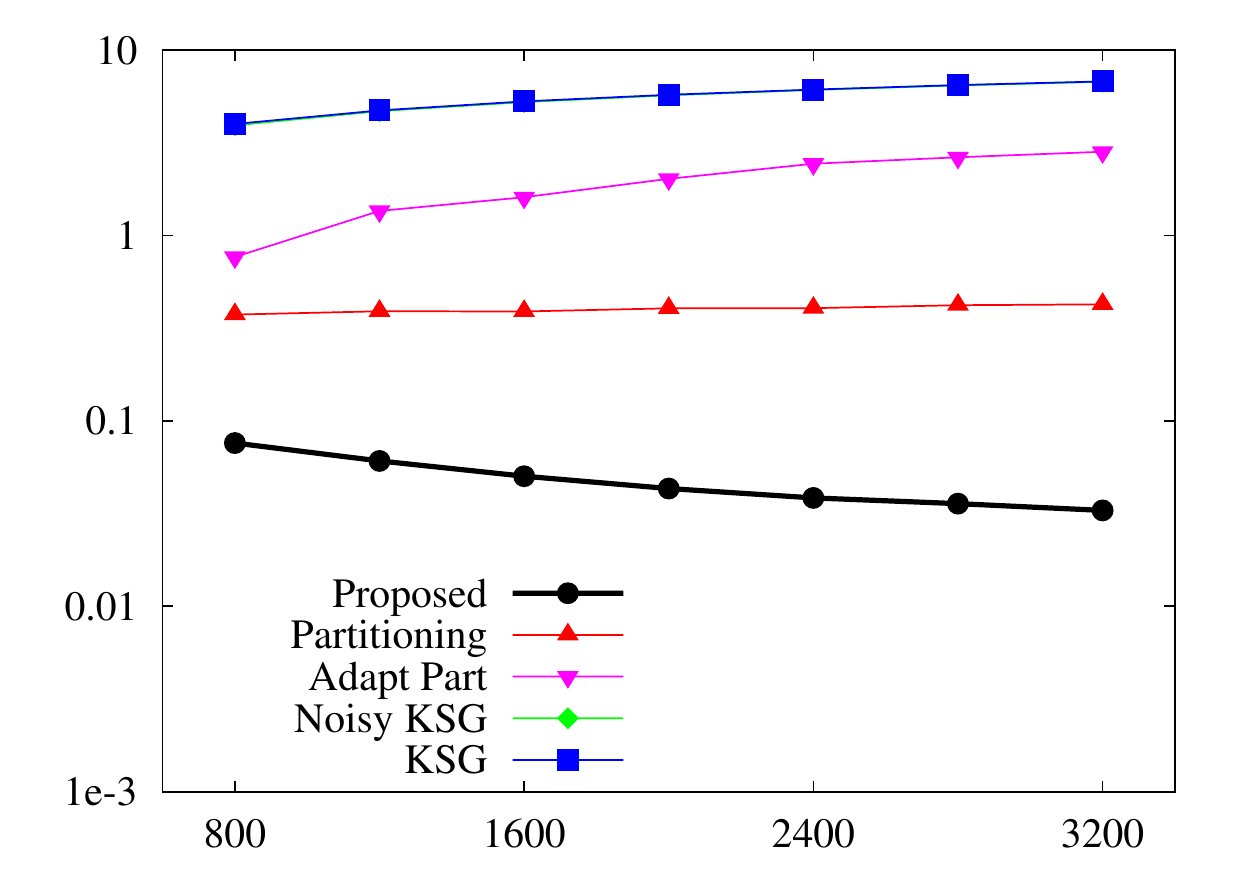}
        \put(-215,30){\rotatebox{90}{mean squared error}}
        \includegraphics[width=.45\textwidth]{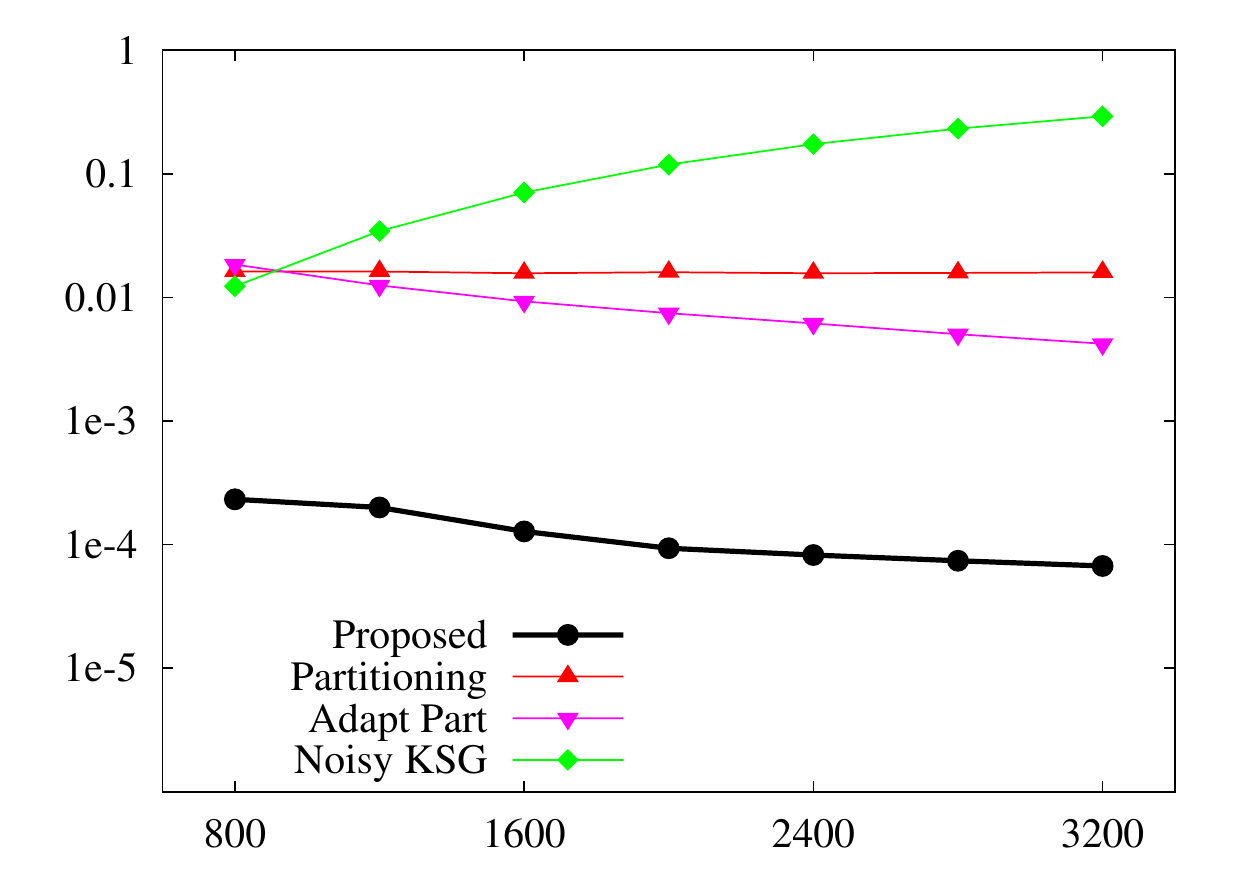}
        \vspace{0.1cm}
        \includegraphics[width=.45\textwidth]{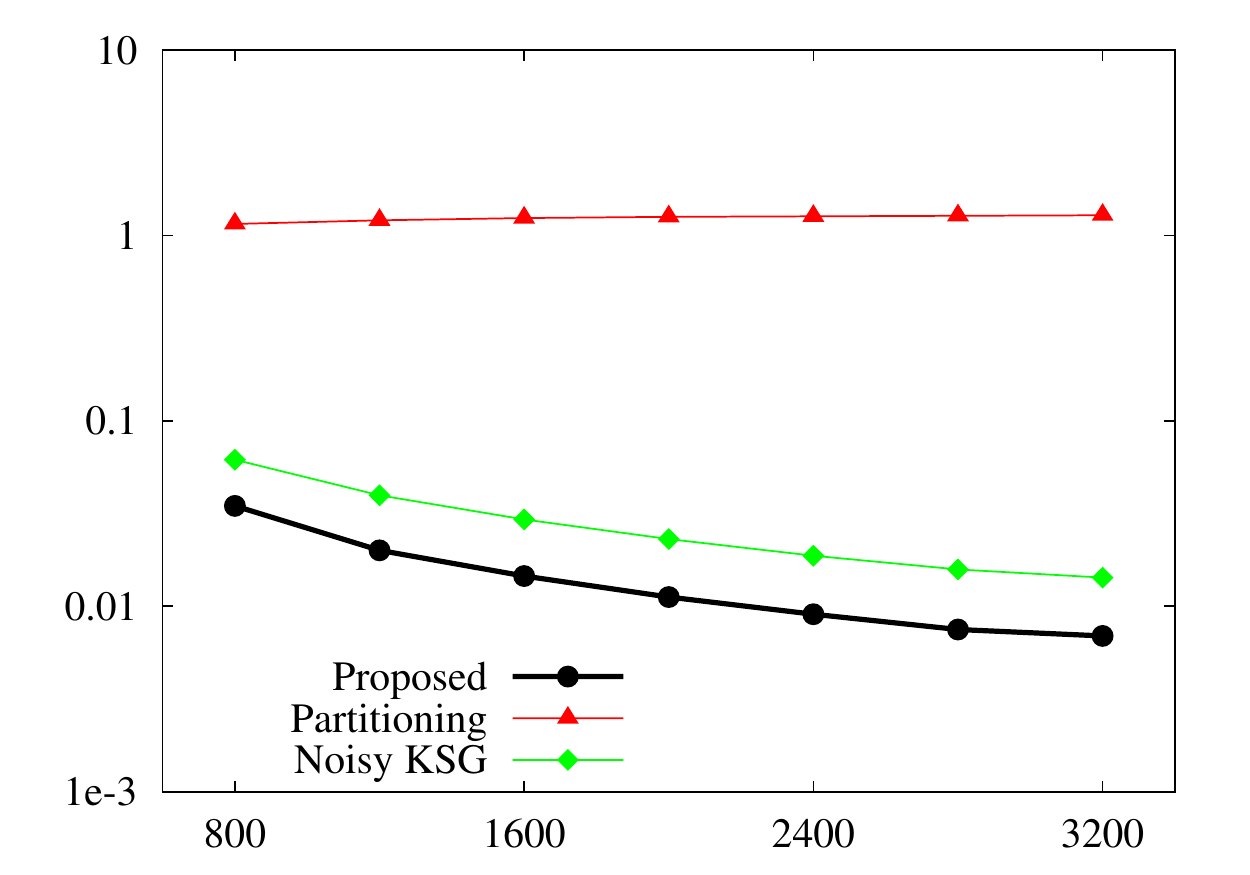}
        \put(-215,30){\rotatebox{90}{mean squared error}}
        \includegraphics[width=.45\textwidth]{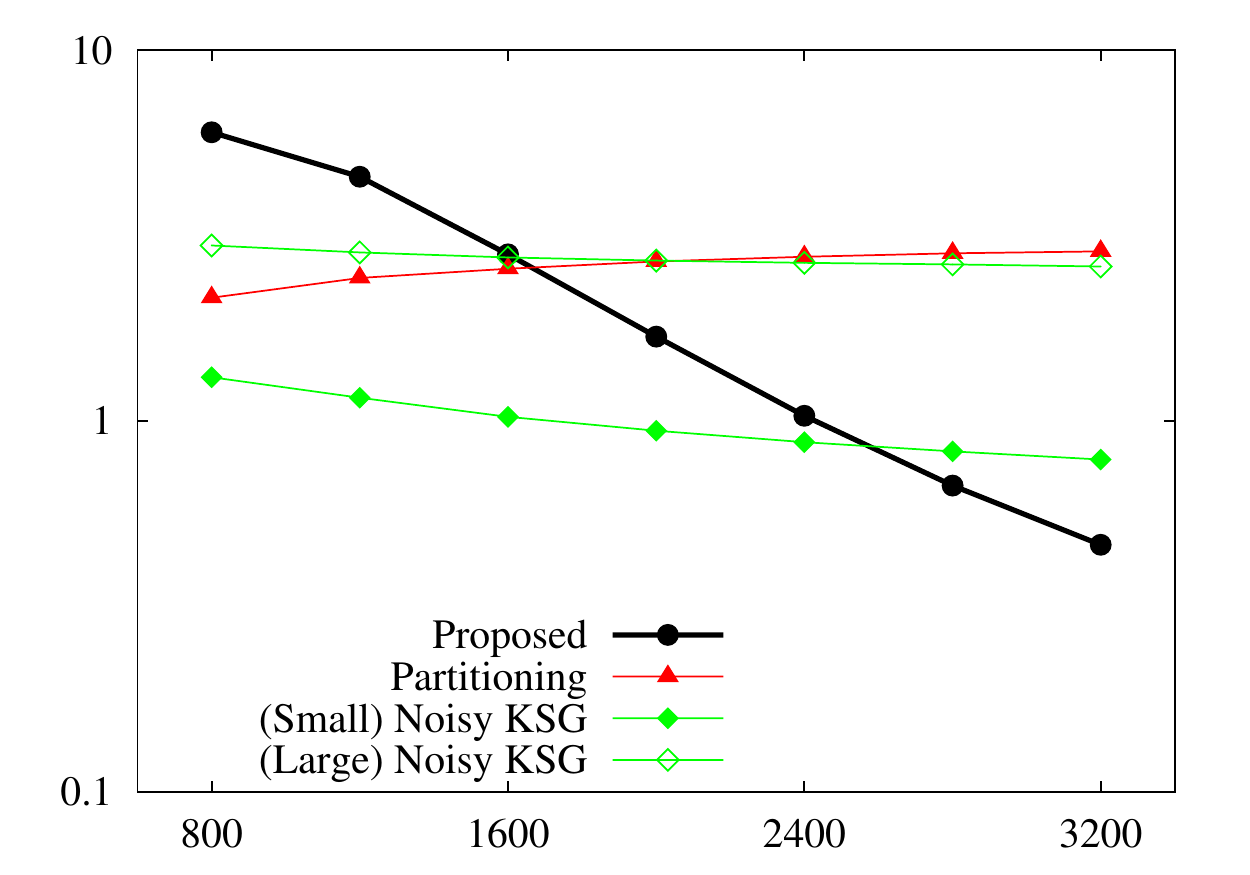}
        \vspace{0.1cm}
        \includegraphics[width=.45\textwidth]{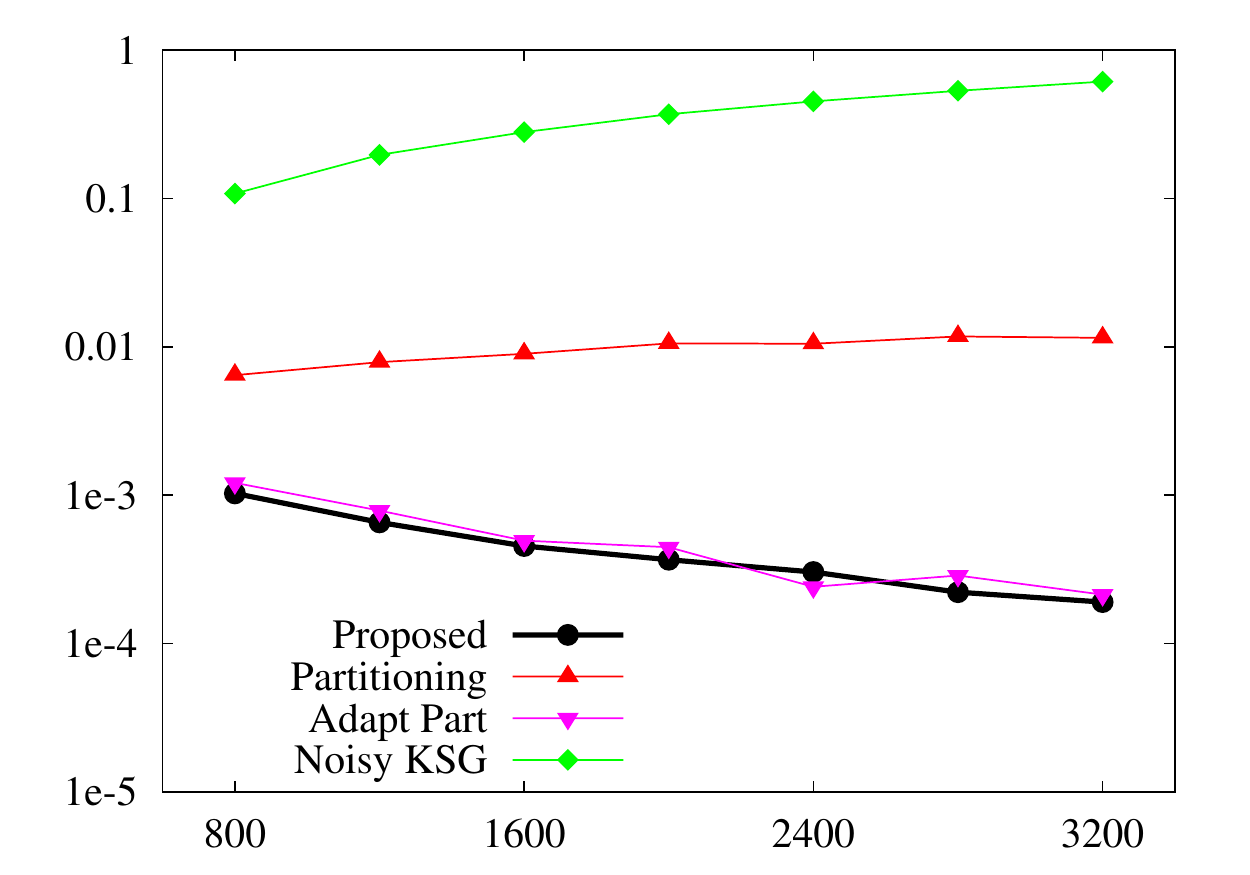}
        \put(-215,30){\rotatebox{90}{mean squared error}}
        \put(-125,-5){sample size}
        \includegraphics[width=.45\textwidth]{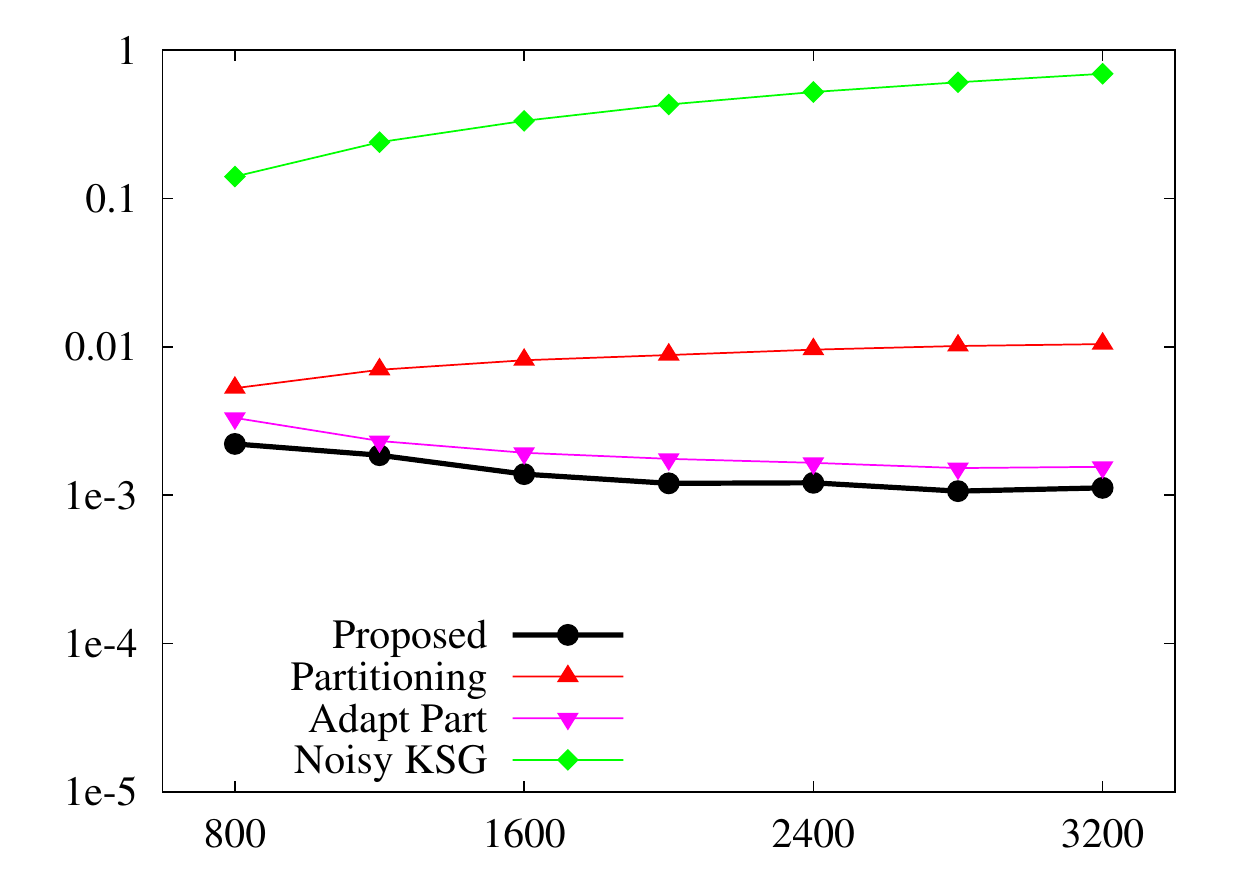}
        \put(-125,-5){sample size}
    \end{center}
    \caption{Mean squared error vs. sample size for synthetic experiments. Top row (left to right): Experiment I; Experiment II. Middle row (left to right): Experiment III for 4 dimensions and 6 dimensions. Bottom row (left to right): Experiment IV for $p = 0$ and $p=15\%$.}
    \label{fig:mse}
\end{figure}

We conclude that our proposed estimator is consistent for all these four experiments, and the mean squared error is always the best or comparable to the best. Other estimators are either not consistent or have large mean squared error for at least one experiment. 

%


\noindent {\bf Feature Selection Task}.
 Suppose there are a set of features modeled by independent random variables $(X_1, \dots, X_p)$ and the data $Y$ depends on a subset of features $\{X_i\}_{i \in S}$, where ${\rm card}(S) = q < p$. We observe the features $(X_1, \dots, X_p)$ and data $Y$ and try to select which features are related to $Y$.
 In many biological applications, some of the data is lost due to experimental reasons and  set to 0; even the available data is noisy. This setting naturally leads to  a mixture of continuous and discrete parts which we model by supposing that the observation is $\tilde{X}_i$ and $\tilde{Y}$, instead of $X_i$ and $Y$. Here $\tilde{X}_i$ and $\tilde{Y}$ equals to 0 with probability $\sigma$ and follows Poisson distribution parameterized by $X_i$ or $Y$ (which corresponds to the noisy observation) with probability $1-\sigma$.

In this experiment, $(X_1, \dots, X_{20})$ are i.i.d.\ standard exponential random variables and $Y$ is simply $(X_1, \dots, X_5)$. $\tilde{X}_i$ equals to 0 with probability 0.15, and $\tilde{X}_i \sim {\rm Poisson}(X_i)$ with probability 0.85. $\tilde{Y}_i$ equals to 0 with probability 0.15 and $\tilde{Y}_i \sim {\rm Exp}(Y_i)$ with probability 0.85. Upon observing $\tilde{X}_i$'s and $\tilde{Y}$, we evaluate ${\rm MI}_i = I(\tilde{X}_i;\tilde{Y})$ using different estimators, and select the features with top-$r$ highest mutual information. Since the underlying number of features is unknown, we iterate over all $r \in \{0, \dots, p\}$ and observe a receiver operating characteristic (ROC) curve, shown in left of Figure~\ref{fig:subset_sel_dream}. Compared to partitioning, noisy KSG and KSG estimators, we conclude that our proposed estimator outperforms other estimators.

\begin{figure}[h]
    \begin{center}
	\includegraphics[width=.45\textwidth]{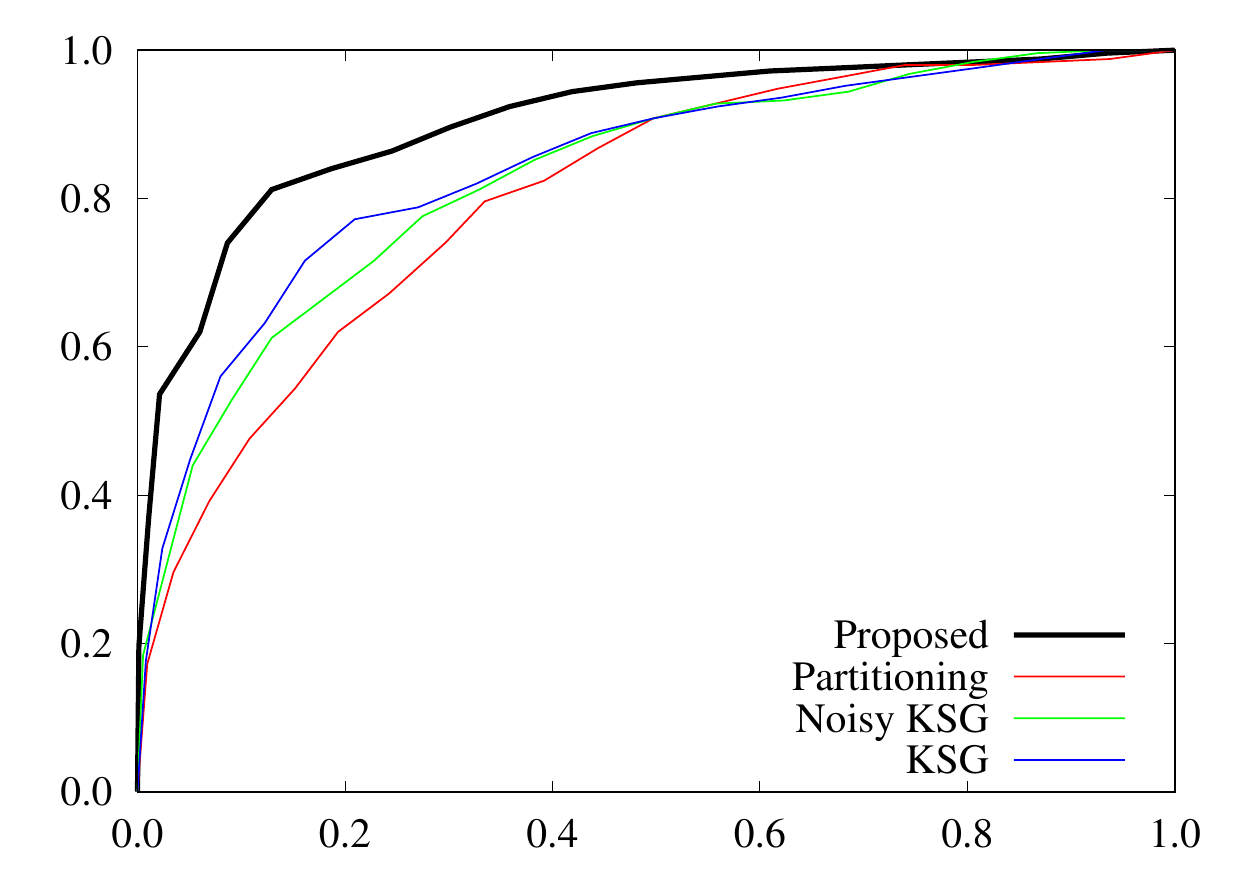}
    \put(-150,-10){False Positive Rate}
    \put(-215,30){\rotatebox{90}{True Positive Rate}}
    \hspace{0.5cm}
    \includegraphics[width=.45\textwidth]{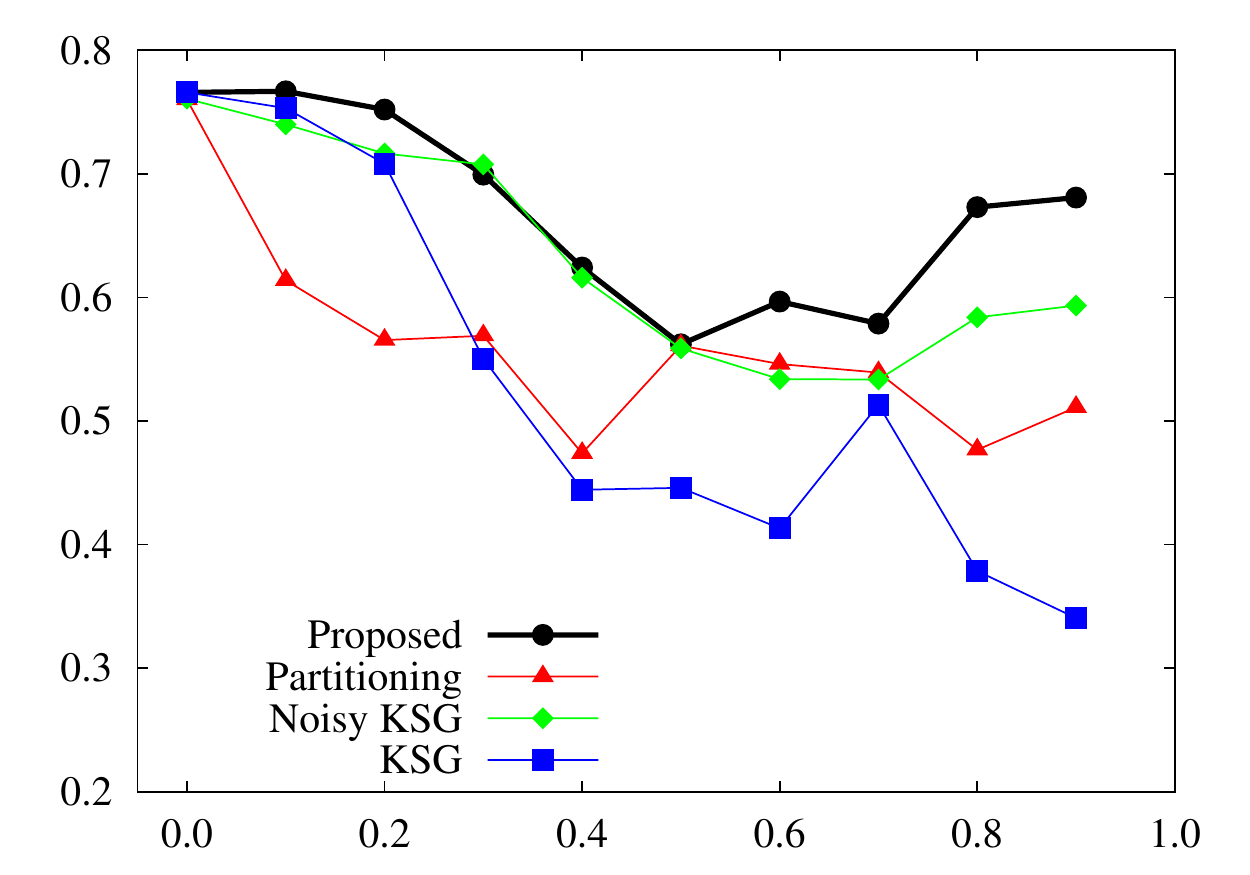}
    \put(-135,-10){Level of Dropout}
    \put(-215,60){\rotatebox{90}{AUROC}}
	\end{center}
	\caption{Left: ROC curve for the feature selection task. Right: AUROC versus levels of dropout for gene regulatory network inference.}
	\label{fig:subset_sel_dream}
\end{figure}

\noindent {\bf Gene regulatory network inference}.
Gene expressions form a rich source of data from which to infer gene regulatory networks; it is now possible to sequence gene expression data from single cells using a technology called single-cell RNA-sequencing \cite{singleCell}. However, this technology has a problem called {\em dropout}, which implies that sometimes, even when the gene is present it is not sequenced \cite{scde,finak2015mast}. While we tested our algorithm on real single-cell RNA-seq dataset, it is hard to establish the ground truth on these datasets. Instead we resorted to a challenge dataset for reconstructing regulatory networks, called the DREAM5 challenge \cite{marbach2012wisdom}. The simulated (insilico) version of this dataset contains gene expression for 20 genes with 660 data point containing various perturbations. The goal is to reconstruct the true network between the various genes. We used mutual information as the test statistic in order to obtain AUROC for various methods. While the dataset did not have any dropouts, in order to simulate the effect of dropouts in real data, we simulated various levels of dropout and compared the AUROC (area under ROC) of different algorithms in the right of Figure~\ref{fig:subset_sel_dream} where we find the proposed algorithm to outperform the competing ones. 

%


\section*{Acknowledgement}
We thank Arman Rahimzamani and Himanshu Asnani for their constructive comments on the proofs of the lemmas, especially for the proof of Lemma~\ref{lem:lemma4}.


\section*{Appendix}
\appendix

%
%

\section{Proof of Theorem~\ref{thm:bias}}
\label{sec:bias}
To prove the asymptotic unbiasedness of the estimator, we need to write the Radon-Nikodym derivative in an explicit form. The following lemma gives the explicit form of $\frac{dP_{XY}}{dP_XP_Y}$.

\begin{lemma}
\label{lem:rnd}
For almost every $(x,y) \in \mathcal{X} \times \mathcal{Y}$, $\frac{dP_{XY}}{dP_XP_Y} = f(x,y) = \lim_{r \to 0} \frac{P_{XY}(x,y,r)}{P_X(x,r)P_Y(y,r)}$.
\end{lemma}

Now notice that $\widehat{I}_N(X;Y) = \frac{1}{N} \sum_{i=1}^N \xi_i$, where all $\xi_i$ are identically distributed. Therefore, $\E[\widehat{I}_N(X;Y)] = \E[\xi_1]$. Therefore, the bias can be written as:
\begin{eqnarray}
\Big|\, \E[\widehat{I}_N(X;Y)] - I(X;Y) \,\Big| &=& \Big|\, \E_{XY} \left[ \E \left[ \xi_1 | X,Y \right] \right] - \int \log f(X,Y) P_{XY} \,\Big| \,\notag\\
&\leq& \int \Big|\, \E \left[ \xi_1| X,Y \right] - \log f(X,Y) \,\Big| \, dP_{XY} \,.
\end{eqnarray}

Now we will give upper bounds for $\Big|\, \E \left[\, \xi_1 | X,Y \,\right] - \log f(X,Y) \,\Big|$ for every $(x,y) \in \mathcal{X} \times \mathcal{Y}$. We will divide the space into three parts as $\mathcal{X} \times \mathcal{Y} = \Omega_1 \bigcup \Omega_2 \bigcup \Omega_3$ where
\begin{itemize}
    \item $\Omega_1 = \{(x,y): f(x,y) = 0\} \,;$
    \item $\Omega_2 = \{(x,y): f(x,y) > 0, P_{XY}(x,y,0) > 0\}\,;$
    \item $\Omega_3 = \{(x,y): f(x,y) > 0, P_{XY}(x,y,0) = 0\} \;.$
\end{itemize}
We will show that $\lim_{N \to \infty} \int_{\Omega_i} \Big|\, \E \left[ \xi_1 | (X,Y)=(x,y)\right] - \log f(x,y) \,\Big| \,dP_{XY} = 0$ for each $i \in \{1,2,3\}$ separately.
\\

$(x,y) \in \Omega_1$: In this case, we will show that $\Omega_1$ has zero probability with respect to $P_{XY}$.
\begin{eqnarray}
P_{XY}(\Omega_1) = \int_{\Omega_1} dP_{XY} = \int_{\Omega_1} f(X,Y) dP_X P_Y = \int_{\Omega_1} 0\, dP_XP_Y = 0
\end{eqnarray}
Therefore, $\int_{\Omega_1} \Big|\, \E \left[ \xi_1 | X,Y \right] - \log f(X,Y) \,\Big| \, dP_{XY} = 0$.
\\

$(x,y) \in \Omega_2$: In this case, $f(x,y)$ is just $P_{XY}(x,y,0)/P_X(x,0)P_Y(y,0)$. We will first show that the probability that the $k$-nearest neighbor distance $\rho_{k,1} > 0$ is small. Then with high probability, we will use the the number of samples on $(x,y)$ as $\tilde{k}_i$, and we will show that the mean of estimate $\xi_1$ is closed to $\log f(x,y)$.
\\

First, the probability of $\rho_{k,1} > 0$ is upper bounded by:
\begin{eqnarray}
&&\Pr \left(\, \rho_{k,1} > 0 \,|\, (X,Y) = (x,y)\,\right) \,\notag\\
&=& \sum_{m=0}^{k-1} {N-1 \choose m} P_{XY}(x,y,0)^m (1-P_{XY}(x,y,0))^{N-1-m} \,\notag\\
&\leq& \sum_{m=0}^{k-1} N^m (1-P_{XY}(x,y,0))^{N-k} \,\notag\\
&\leq& k N^k (1-P_{XY}(x,y,0))^{N-k} \,\notag\\
&\leq& kN^ke^{-(N-k)P_{XY}(x,y,0)} \;.
\end{eqnarray}

Conditioning on the event that $\rho_{k,1} = 0$, we have $\xi_1 = \psi(\tilde{k}_1) + \log N - \log(n_{x,1}+1) - \log(n_{y,1}+1)$. Then we write $\Big|\, \E \left[ \xi_1 | (X,Y)=(x,y), \rho_{k,1} = 0 \right] - \log f(x,y) \,\Big|$ as
\begin{eqnarray}
&&\Big|\, \E \left[ \xi_1 | (X,Y)=(x,y), \rho_{k,1} = 0 \right] - \log f(x,y) \,\Big| \,\notag\\
&=& \Big|\, \E \left[ \psi(\tilde{k}_1) + \log N - \log(n_{x,1}+1) - \log(n_{y,1}+1) | (X,Y)=(x,y), \rho_{k,1} = 0 \right] \,\notag\\
&&-\, \log \frac{P_{XY}(x,y,0)}{P_X(x,0)P_Y(y,0)} \,\Big| \,\notag\\
&\leq&  \Big|\, \E \left[ \log(n_{x,1}+1) | (X,Y)=(x,y), \rho_{k,1} = 0 \right] - \log NP_X(x,0) \,\Big| \,\notag\\
&&+\, \Big|\, \E \left[ \log(n_{y,1}+1) | (X,Y)=(x,y), \rho_{k,1} = 0 \right] - \log NP_Y(y,0) \,\Big| \,\notag\\
&&+\, \Big|\, \E \left[ \psi(\tilde{k}_1) | (X,Y)=(x,y), \rho_{k,1} = 0 \right] - \log NP_{XY}(x,y,0) \,\Big| \label{eq:caseii}
\end{eqnarray}

Notice that $\tilde{k}_1$ is the number of samples among $\{(X_i,Y_i)\}_{i=2}^N$ such that $(X_i, Y_i) = (x,y)$, where each $(X_i, Y_i) = (x,y)$ with probability $P_{XY}(x,y,0)$. Therefore, the distribution of $\tilde{k}_1$ is ${\rm Bino}(N-1, P_{XY}(x,y,0))$. Similarly, $n_{x,1}$ is the number of samples among $\{(X_i,Y_i)\}_{i=2}^N$ such that $X_i = x$, $n_{y,1}$ is the number of samples among $\{(X_i,Y_i)\}_{i=2}^N$ such that $Y_i = y$. Therefore, $n_{x,1} \sim {\rm Bino}(N-1, P_X(x,0))$ and $n_{y,1} \sim {\rm Bino}(N-1, P_Y(y,0))$. Notice that conditioning on $\rho_{k,i}=0$ is equivalent to conditioning on $\tilde{k}_i \geq k$, or $n_{x,i} \geq k$, $n_{y,i} \geq k$, so we propose the following lemma to deal with~\eqref{eq:caseii}.

\begin{lemma}
\label{lem:lemma4}
If $X$ is distributed as $\text{Bino}(N,p)$ and $m \geq 0$ , then:
\begin{eqnarray}
\left| \mathbb{E} \left[ \log(X+m) \middle| X \geq k \right] - \log (Np) \right| \leq \max \left\{ \left| \log \left( \frac{1+\frac{m}{Np}}{1 - \exp\left( -2 \frac{(Np-k)^2}{N} \right)} \right) \right|
 , \frac{1}{1 - \exp\left( -2 \frac{(Np-k)^2}{N} \right)} \frac{3}{2Np} \right\}.
\end{eqnarray}
\end{lemma}

By Assumption 2, $k/N \rightarrow 0$ as $N \rightarrow \infty$, then $(Np-k)^2/N= N (p-k/N)^2 \rightarrow \infty$, So for sufficiently large $N$, the RHS of Lemma~\ref{lem:lemma4} is upper bounded by $\max\{\frac{C_1 m}{Np}, \frac{C_2}{Np}\} \leq \frac{C(m+1)}{Np}$, where $C = \max\{C_1, C_2\}$ is some constant not depends on $N$. Therefore, by applying Lemma~\ref{lem:lemma4} with $m=1$, the first term of~\eqref{eq:caseii} is bounded by:
\begin{eqnarray}
&&\Big|\, \E \left[ \log(n_{x,1}+1) | (X,Y)=(x,y), \rho_{k,1} = 0 \right] - \log NP_X(x,0) \,\Big| \,\notag\\
&\leq& \Big|\, \E \left[ \log(n_{x,1}+1) | (X,Y)=(x,y), n_{x,i} \geq k \right] - \log (N-1)P_X(x,0) \,\Big| + \log \frac{N}{N-1} \,\notag\\
&\leq& \frac{2C}{(N-1)P_X(x,0)} + \frac{1}{N-1} \,\notag\\
&\leq& \frac{2C+1}{(N-1)P_X(x,0)} \leq \frac{4C+2}{NP_X(x,0)} \;.
\end{eqnarray}
Similarly, the second term of~\eqref{eq:caseii} is bounded by: $(4C+2)/(NP_Y(y,0))$. For the third term, notice that $|\psi(x) - \log(x)| \leq 1/x$ for every integer $ x \geq 1$, therefore, $|\psi(\tilde{k}_1) - \log(\tilde{k}_1)| \leq 1/\tilde{k}_1 \leq 1/k$. By applying Lemma~\ref{lem:lemma4} with $m=0$, the third term of ~\eqref{eq:caseii} is bounded by: $(2C+2)/(NP_{XY}(x,y,0))+1/k$. By Combining three terms together and noticing that $P_X(x,0) \geq P_{XY}(x,y,0)$ and $P_Y(y,0) \geq P_{XY}(x,y,0)$, we obtain
\begin{eqnarray}
&&\Big|\, \E \left[ \xi_1 | (X,Y)=(x,y), \rho_{k,1} = 0 \right] - \log f(x,y) \,\Big| \,\notag\\
&\leq& \frac{4C+2}{NP_X(x,0)} + \frac{4C+2}{NP_Y(y,0)} + \frac{2C+2}{NP_{XY}(x,y,0)} + \frac{1}{k} \leq \frac{10C+6}{NP_{XY}(x,y,0)} + \frac{1}{k} \;.
\end{eqnarray}
Combine with the case that $\rho_{i,xy} > 0$, we obtain that:
\begin{eqnarray}
&&\Big|\, \E \left[ \xi_1 | (X,Y)=(x,y)\right] - \log f(x,y) \,\Big| \,\notag\\
&\leq& \Big|\, \E \left[ \xi_1 | (X,Y)=(x,y), \rho_{k,1} > 0 \right] - \log f(x,y) \,\Big| \times \Pr \left(\, \rho_{k,1} > 0 \,\right) \,\notag\\
&&+\, \Big|\, \E \left[ \xi_1 | (X,Y)=(x,y), \rho_{k,1} = 0 \right] - \log f(x,y) \,\Big| \times \Pr \left(\, \rho_{k,1} = 0 \,\right) \,\notag\\
&\leq& (2 \log N + |\, \log f(x,y) \,|) k N^k e^{-(N-k)P_{XY}(x,y,0)} + \frac{10C+6}{NP_{XY}(x,y,0)} + \frac{1}{k} \;,
\end{eqnarray}
where the first term comes from triangle inequality and the fact that $|\xi_1| \leq 2 \log N$. Integrating over $\Omega_2$, we have:
\begin{eqnarray}
&&\int_{\Omega_2} \Big|\, \E \left[ \xi_1 | (X,Y)=(x,y)\right] - \log f(x,y) \,\Big| \,dP_{XY} \,\notag\\
&\leq& \int_{\Omega_2} (2 \log N + |\, \log f(x,y) \,|) k N^k e^{-(N-k)P_{XY}(x,y,0)} \, dP_{XY} \,\notag\\
&&+\, \frac{10C+6}{N} \int_{\Omega_2} \frac{1}{P_{XY}(x,y,0)} dP_{XY}+ \frac{1}{k} \,\notag\\
&\leq& (2 \log N + \int_{\Omega_2} |\, \log f(x,y) \,| dP_{XY}) k N^k e^{-(N-k) \inf_{(x,y)\in \Omega_2} P_{XY}(x,y,0)} \,\notag\\
&&+\, \frac{10C+6}{N} \mu(\Omega_2) + \frac{1}{k} \;,
\end{eqnarray}
where $\mu$ denotes counting measure. By Assumption 1, $k$ goes to infinity as $N$ goes to infinity, so $1/k$ vanishes as $N$ increases. By Assumption 1 and 2, $k/N$ goes to 0 and $\Omega_2$ has finite counting measure, so the second term also vanishes. Since $\Omega_2$ has finite counting measure, so $\inf_{(x,y)\in \Omega_2} P_{XY}(x,y,0) = \epsilon > 0$. By Assumption 3, $\int_{\Omega_2} |\, \log f(x,y) \,| dP_{XY} < +\infty$. Therefore, for sufficiently large $N$, the first term also vanishes. Therefore,
\begin{eqnarray}
\lim_{N \to \infty} \int_{\Omega_2} \Big|\, \E \left[ \xi_1 | (X,Y)=(x,y)\right] - \log f(x,y) \,\Big| \,dP_{XY} = 0 \;.
\end{eqnarray}

$(x,y) \in \Omega_3$: In this case, $P_{XY}(x,y,r)$ is a monotonic function of $r$ such that $P_{XY}(x,y,0) = 0$ and $\lim_{r \to \infty} P_{XY}(x,y,r) = 1$. Hence, we can view $\log \left(\, P_{XY}(x,y,r)/P_X(x,r)P_Y(y,r) \,\right)$ as a function of $P_{XY}(x,y,r)$, and it converges to $\log f(x,y)$ as $P_{XY}(x,y,r) \to 0$, for almost every $(x,y)$. Since $P_{XY}(\Omega_3) \leq 1 < +\infty$ and $\int_{\Omega_3} |\log f(x,y)| dP_{XY} < +\infty$. Then by Egoroff's Theorem, for any $\epsilon > 0$, there exists a subset $E \subseteq \Omega_3$ with $P_{XY}(E) < \epsilon$ and $\int_E |\log f(x,y)| dP_{XY} < \epsilon$, such that $\log \left(\, P_{XY}(x,y,r)/P_X(x,r)P_Y(y,r) \,\right)$ converges as $P_{XY}(x,y,r) \to 0$, uniformly on $\Omega_3 \setminus E$. For $(x, y) \in E$, notice that $|\xi_1| \leq 2 \log N$, so we have:
\begin{eqnarray}
&&\, \int_{E} \Big|\, \E \left[ \xi_1 | (X,Y)=(x,y)\right] - \log f(x,y) \,\Big| \,dP_{XY} \,\notag\\
&\leq& \int_{E} \left(\, 2\log N + |\, \log f(x,y) \,| \,\right) \,dP_{XY} < (2 \log N + 1)\epsilon \;.
\end{eqnarray}
By choosing $\epsilon$ appropriately, we will have $\lim_{N \to \infty} \int_{E} \Big|\, \E \left[ \xi_1 | (X,Y)=(x,y)\right] - \log f(x,y) \,\Big| \,dP_{XY} = 0$.
\\

Now for any $(x, y) \in \Omega_3 \setminus E$, since $P_{XY}(x,y,0) = 0$, we know that $\Pr \left(\, \rho_{k,1} = 0 \,|\, (X,Y) = (x,y)  \,\right) = 0$, so $\tilde{k}_1 = k$ with probability $1$. Conditioning on $\rho_{k,1} = r > 0$, the difference $\Big|\, \E \left[ \xi_1 | (X,Y)=(x,y)\right] - \log f(x,y) \,\Big|$ can be decomposed into four parts as follows
\begin{eqnarray}
&&\Big|\, \E \left[ \xi_1 | (X,Y)=(x,y)\right] - \log f(x,y) \,\Big| \,\notag\\
&=& \Big|\, \int_{r=0}^{\infty} \left(\, \E \left[ \xi_1 | (X,Y)=(x,y), \rho_{k,1} = r\right] - \log f(x,y) \,\right) dF_{\rho_{k,1}}(r) \,\Big| \,\notag\\
&\leq& \Big|\, \int_{r=0}^{\infty} \left(\, \log \frac{P_{XY}(x,y,r)}{P_X(x,r)P_Y(y,r)} - \log f(x,y) \,\right) dF_{\rho_{k,1}}(r) \,\Big| \,\label{eq:i4}\\
&&+\, \Big|\, \int_{r=0}^{\infty} \left(\, \psi(k) - \log N - \log P_{XY}(x,y,r) \,\right) dF_{\rho_{k,1}}(r) \,\Big| \,\label{eq:i1}\\
&&+\, \Big|\, \int_{r=0}^{\infty} \left(\, \E \left[ \log(n_{x,1}+1) | (X,Y,\rho_{k,1})=(x,y,r)\right] - \log (NP_X(x,r)) \,\right) dF_{\rho_{k,1}}(r) \,\Big| \,\label{eq:i2}\\
&&+\, \Big|\, \int_{r=0}^{\infty} \left(\, \E \left[ \log(n_{y,1}+1) | (X,Y,\rho_{k,1})=(x,y,r)\right] - \log (NP_Y(y,r)) \,\right) dF_{\rho_{k,1}}(r) \,\Big| \,\label{eq:i3}
\end{eqnarray}
here $F_{\rho_{k,1}}(r)$ is the CDF of the $k$-nearest neighbor distance $\rho_{k,1}$, given $(X,Y) = (x,y)$. By results of order statistics, its derivative with respect to $P_{XY}(x,y,r)$ is given by:
\begin{eqnarray}
\frac{d F_{\rho_{k,1}}(r)}{d P_{XY}(x,y,r)} &=& \frac{(N-1)!}{(k-1)!(N-k-1)!} P_{XY}(x,y,r)^{k-1} \left(\, 1-P_{XY}(x,y,r) \,\right)^{N-k-1}\;. \label{eq:f_rho}
\end{eqnarray}

Now we consider the four terms separately. For~\eqref{eq:i4}, since $\log \left(\, P_{XY}(x,y,r)/P_X(x,r)P_Y(y,r) \,\right)$ converges as $P_{XY}(x,y,r) \to 0$, uniformly on $\Omega_3 \setminus E$. So for every $(x,y) \in \Omega_3 \setminus E$, there exists an $r_N$ such that $P_{XY}(x,y,r_N) = 4k\log N/N$ and $|\log \left(\, P_{XY}(x,y,r)/P_X(x,r)P_Y(y,r) \,\right) - \log f(x,y)| < \delta_N$ for every $r \leq r_N$. Here $r_N$ may depend on $(x,y)$, but $\delta_N$ does not depend on $(x,y)$ and $\lim_{N \to \infty} \delta_N = 0$. Therefore,~\eqref{eq:i4} is upper bounded by:
\begin{eqnarray}
&&\Big|\, \int_{r=0}^{\infty} \left(\, \log \frac{P_{XY}(x,y,r)}{P_X(x,r)P_Y(y,r)} - \log f(x,y) \,\right) dF_{\rho_{k,1}}(r) \,\Big| \,\notag\\
&\leq& \int_{r=0}^{r_N} \Big|\, \log \frac{P_{XY}(x,y,r)}{P_X(x,r)P_Y(y,r)} - \log f(x,y) \,\Big| dF_{\rho_{k,1}}(r) \,\notag\\
&&+\,  \int_{r=r_N}^{\infty} \Big|\, \log \frac{P_{XY}(x,y,r)}{P_X(x,r)P_Y(y,r)} - \log f(x,y) \,\Big| dF_{\rho_{k,1}}(r) \,\notag\\
&\leq& \delta_N \Pr \left(\, \rho_{k,1} \leq r_N \,|\, (X,Y) = (x,y) \,\right) \,\notag\\
&&+\, \left(\, \sup_{r \geq r_N} \Big|\, \log \frac{P_{XY}(x,y,r)}{P_X(x,r)P_Y(y,r)} - \log f(x,y) \,\Big| \,\right) \Pr \left(\, \rho_{k,1} > r_N \,|\, (X,Y) = (x,y) \,\right) \;.
\end{eqnarray}
Firstly, the probability $\Pr \left(\, \rho_{k,1} \leq r_N \,|\, (X,Y) = (x,y) \,\right)$ is smaller than 1. Secondly, since $P_X(x,y,r) \geq 4k \log N/N > 1/N$ for $r \geq r_N$, so we have $|\log P_{XY}(x,y,r)| \leq \log N $. The same bounds apply for $|\log P_X(x,r)|$ and $|\log P_Y(y,r)|$ as well. By triangle inequality, the supremum is upper bounded by $3 \log N + |\log f(x,y)|$. Finally, the probability $\Pr \left(\, \rho_{k,1} > r_N \,|\, (X,Y) = (x,y) \,\right)$ is upper bounded by
\begin{eqnarray}
&&\Pr \left(\, \rho_{k,1} > r_N \,|\, (X,Y) = (x,y)\,\right) \,\notag\\
&=& \sum_{m=0}^{k-1} {N-1 \choose m} P_{XY}(x,y,r_N)^m (1-P_{XY}(x,y,r_N))^{N-1-m} \,\notag\\
&\leq& \sum_{m=0}^{k-1} N^m (1-P_{XY}(x,y,r_N))^{N-k} \,\notag\\
&=& k N^k (1-\frac{4k\log N}{N})^{N/2} \,\notag\\
&\leq& kN^k e^{-2k \log N} = \frac{k}{N^k} \;.
\end{eqnarray}
for sufficiently large $N$ such that $N-k > N/2$. Therefore,~\eqref{eq:i4} is upper bounded by
\begin{eqnarray}
&&\,\Big|\, \int_{r=0}^{\infty} \left(\, \log \frac{P_{XY}(x,y,r)}{P_X(x,r)P_Y(y,r)} - \log f(x,y) \,\right) dF_{\rho_{k,1}}(r) \,\Big| \,\notag\\
&\leq& \delta_N + \frac{k(3 \log N + |\log f(x,y)|)}{N^k}\;.~\label{eq:ub_i4}
\end{eqnarray}

For~\eqref{eq:i1}, we simply plug in $F_{\rho_{k,1}}(r)$ and integrate over $P_{XY}(x,y,r)$ and obtain
\begin{eqnarray}
&& \int_{r=0}^{\infty} \left(\, \psi(k) - \log N - \log P_{XY}(x,y,r) \,\right) dF_{\rho_{k,1}}(r) \,\notag\\
&=& \psi(k) - \log N - \frac{(N-1)!}{(k-1)!(N-k-1)!} \,\notag\\
&&\times\, \int_{r=0}^{\infty} (\log P_{XY}(x,y,r)) P_{XY}(x,y,r)^{k-1} \left(\, 1-P_{XY}(x,y,r) \,\right)^{N-k-1} d P_{XY}(x,y,r) \,\notag\\
&=& \psi(k) - \log N - \frac{(N-1)!}{(k-1)!(N-k-1)!} \int_{t=0}^1 (\log t) t^{k-1} (1-t)^{N-k-1} dt \,\notag\\
&=& \psi(k) - \log N - (\psi(k) - \psi(N)) = \psi(N) - \log N \;. \label{eq:ub_i1}
\end{eqnarray}
where we use the fact that $\psi(k) - \psi(N) = \frac{(N-1)!}{(k-1)!(N-k-1)!} \int_{t=0}^1 (\log t) t^{k-1} (1-t)^{N-k-1} dt$. Notice that $\psi(N) < \log N$ and $\lim_{N \to 0} (\psi(N) - \log N) = 0$.

Now we deal with~\eqref{eq:i2} and~\eqref{eq:i3}. The following lemmas establish the distribution of $n_{x,1}$ and $n_{y,1}$ given $(X,Y)=(x,y)$ and $\rho_{k,1} = r > 0$.
\begin{lemma}
\label{lem:bino}
Given $(X,Y) = (x,y)$ and $\rho_{k,1} = r > 0$, then $n_{x,1} - k$ is distributed as ${\rm Bino}(N-k-1, \frac{P_X(x,r)-P_{XY}(x,y,r)}{1-P_{XY}(x,y,r)})$; $n_{y,1} - k$ is distributed as ${\rm Bino}(N-k-1, \frac{P_Y(y,r)-P_{XY}(x,y,r)}{1-P_{XY}(x,y,r)})$.
\end{lemma}

The following lemma is useful to establish the upper bound for~\eqref{eq:i2} and~\eqref{eq:i3}.
\begin{lemma}
\label{lem:elog}
For integer $m \geq 1$, if $X$ is distributed as ${\rm Bino}(N,p)$, then $|\E [\log (X+m)] - \log(Np+m)| \leq C/(Np+m)$ for some constant $C$.
\end{lemma}

Now we are ready to upper bound~\eqref{eq:i2}. First, we rewrite the term~\eqref{eq:i2} as:
\begin{eqnarray}
&&\Big|\, \int_{r=0}^{\infty} \left(\, \E \left[ \log(n_{x,1}+1) | (X,Y)=(x,y), \rho_{k,1} = r\right] - \log N - \log P_X(x,r) \,\right) dF_{\rho_{k,1}}(r) \,\Big|\,\notag\\
&\leq&\Big|\, \int_{r=0}^{\infty} \Big(\, \E \left[ \log(n_{x,1}+1) | (X,Y)=(x,y), \rho_{k,1} = r\right] \,\notag\\
&&-\, \log \left(\, (N-k-1)\frac{P_X(x,r)-P_{XY}(x,y,r)}{1-P_{XY}(x,y,r)}+k+1 \,\right) \,\Big) dF_{\rho_{k,1}}(r) \,\Big|\,\notag\\
&&+\, \Big|\, \int_{r=0}^{\infty} \left(\, \log \frac{ (N-k-1)\frac{P_X(x,r)-P_{XY}(x,y,r)}{1-P_{XY}(x,y,r)}+k+1}{NP_X(x,r)} \,\right) dF_{\rho_{k,1}}(r) \,\Big|\,\notag\\
&\leq& \int_{r=0}^{\infty} \Big|\, \E \left[ \log(n_{x,1}+1) | (X,Y)=(x,y), \rho_{k,1} = r\right] \,\notag\\
&&-\, \log \left(\, (N-k-1)\frac{P_X(x,r)-P_{XY}(x,y,r)}{1-P_{XY}(x,y,r)}+k+1 \,\right) \,\Big| dF_{\rho_{k,1}}(r) \,\label{eq:i2_1}\\
&&+\,\Big|\, \E_r \left[\, \log \left(\, \frac{N(P_X(x,r)-P_{XY}(x,y,r))+(k+1)(1-P_X(x,r))}{NP_X(x,r)(1-P_{XY}(x,y,r))} \,\right) \,\right] \,\Big|\;.\label{eq:i2_2}
\end{eqnarray}
where $\E_r$ denotes expectation over $F_{\rho_{i,xy}}$. By Lemma~\ref{lem:elog}, the term~\eqref{eq:i2_1} is upper bounded by
\begin{eqnarray}
&&\int_{r=0}^{\infty} \Big|\, \E \left[ \log(n_{x,1}+1) | (X,Y)=(x,y), \rho_{k,1} = r\right] \,\notag\\
&&-\, \log \left(\, (N-k-1)\frac{P_X(x,r)-P_{XY}(x,y,r)}{1-P_{XY}(x,y,r)}+k+1 \,\right) \,\Big| dF_{\rho_{k,1}}(r) \,\notag\\
&\leq& \int_{r=0}^{\infty} \frac{C}{(N-k-1)\frac{P_X(x,r)-P_{XY}(x,y,r)}{1-P_{XY}(x,y,r)}+k+1} dF_{\rho_{k,1}}(r) \,\notag\\
&\leq& \int_{r=0}^{\infty} \frac{C}{k+1} dF_{\rho_{k,1}}(r) = \frac{C}{k+1}\;. \label{eq:ub_i21}
\end{eqnarray}
For~\eqref{eq:i2_2}, by the fact that $\log(x/y) \leq (x-y)/y$ for all $x, y > 0$ and Cauchy-Schwarz inequality, we have the following:
\begin{eqnarray}
&&\E_r \left[\, \log \left(\, \frac{N(P_X(x,r)-P_{XY}(x,y,r))+(k+1)(1-P_X(x,r))}{NP_X(x,r)(1-P_{XY}(x,y,r))} \,\right) \,\right] \,\notag\\
&\leq& \E_r \left[\, \frac{N(P_X(x,r)-P_{XY}(x,y,r))+(k+1)(1-P_X(x,r))}{NP_X(x,r)(1-P_{XY}(x,y,r))}-1 \,\right] \,\notag\\
&=& \E_r \left[\, \frac{(k+1-NP_{XY}(x,y,r))(1-P_X(x,r))}{NP_X(x,r)(1-P_{XY}(x,y,r))}\,\right] \,\notag\\
&\leq& \sqrt{\E_r \left[\, \left(\frac{k+1-NP_{XY}(x,y,r)}{NP_{XY}(x,y,r)}\right)^2 \,\right] \E_r \left[\, \left( \frac{P_{XY}(x,y,r)(1-P_X(x,r))}{P_X(x,r)(1-P_{XY}(x,y,r))}\right)^2 \,\right] } \;.
\end{eqnarray}
Notice that $P_X(x,r) \geq P_{XY}(x,y,r)$ for all $r$, so the second expectation is always no larger than 1. For the first expectation, we plug in $F_{\rho_{k,1}}(r)$ and integrate over $P_{XY}(x,y,r)$, let $t = P_{XY}(x,y,r)$ and observe,
\begin{eqnarray}
&&\E_r \left[\, \left(\frac{k+1-NP_{XY}(x,y,r)}{NP_{XY}(x,y,r)}\right)^2 \,\right] \,\notag\\
&=& \int_{r=0}^{\infty} \left(\frac{k+1-NP_{XY}(x,y,r)}{NP_{XY}(x,y,r)}\right)^2 dF_{\rho_{i,xy}}(r) \,\notag\\
&=& \frac{(N-1)!}{(k-1)!(N-k-1)!} \int_{t=0}^1 \frac{(k+1-Nt)^2}{N^2t^2} t^{k-1} (1-t)^{N-k-1} dt \,\notag\\
&=& \frac{(N-1)!}{(k-1)!(N-k-1)!}\frac{(k+1)^2}{N^2} \int_{t=0}^1 t^{k-3}(1-t)^{N-k-1} dt \,\notag\\
&&-\, \frac{(N-1)!}{(k-1)!(N-k-1)!}\frac{2(k+1)}{N^2} \int_{t=0}^1 t^{k-2}(1-t)^{N-k-1} dt \,\notag\\
&&+\, \frac{(N-1)!}{(k-1)!(N-k-1)!} \int_{t=0}^1 t^{k-3}(1-t)^{N-k-1} dt \,\notag\\
&=& \frac{(N-1)!}{(k-1)!(N-k-1)!}\frac{(k+1)^2}{N^2} \frac{(k-3)!(N-k-1)!}{(N-3)!} \,\notag\\
&&-\, \frac{(N-1)!}{(k-1)!(N-k-1)!}\frac{2(k+1)}{N^2} \frac{(k-2)!(N-k-1)!}{(N-2)!} + 1 \,\notag\\
&=& \frac{(N-1)(N-2)(k+1)^2}{N^2(k-1)(k-2)} - \frac{2(N-1)(k+1)}{N(k-1)} + 1 \;.
\end{eqnarray}
For sufficiently large $N$ and $k$, it is upper bounded by $C_1(1/N+1/k)$ for some constant $C_1 > 0$. Therefore,
\begin{eqnarray}
\E_r \left[\, \log \left(\, \frac{N(P_X(x,r)-P_{XY}(x,y,r))+(k+1)(1-P_X(x,r))}{NP_X(x,r)(1-P_{XY}(x,y,r))} \,\right) \,\right] \leq \sqrt{C_1(\frac{1}{N} + \frac{1}{k})} \;. \label{eq:ub_i22}
\end{eqnarray}
Similarly, by using the fact that $\log(x/y) > (x-y)/x$ and Cauchy-Schwarz inequality again, we conclude that there are some constant $C_2 > 0$ such that
\begin{eqnarray}
\E_r \left[\, \log \left(\, \frac{N(P_X(x,r)-P_{XY}(x,y,r))+(k+1)(1-P_X(x,r))}{NP_X(x,r)(1-P_{XY}(x,y,r))} \,\right) \,\right] \geq -\sqrt{C_2(\frac{1}{N} + \frac{1}{k})} \;. \label{eq:ub_i23}
\end{eqnarray}
Therefore, by combining~\eqref{eq:ub_i21},~\eqref{eq:ub_i22} and~\eqref{eq:ub_i23}, we obtain
\begin{eqnarray}
&&\Big|\, \int_{r=0}^{\infty} \left(\, \E \left[ \log(n_{x,1}+1) | (X,Y)=(x,y), \rho_{k,1} = r\right] - \log N - \log P_X(x,r) \,\right) dF_{\rho_{k,1}}(r) \,\Big| \,\notag\\
&\leq& \frac{C}{k+1} + \sqrt{C'(\frac{1}{N} + \frac{1}{k})} \;.\label{eq:ub_i2}
\end{eqnarray}
where $C' = \max\{C_1, C_2\}$. Since~\eqref{eq:i3} and~\eqref{eq:i2} are symmetric, the same upper bound~\eqref{eq:ub_i2} also applies to~\eqref{eq:i3}. Combine~\eqref{eq:ub_i4},~\eqref{eq:ub_i1} and~\eqref{eq:ub_i2}, we have
\begin{eqnarray}
&&\Big|\, \E \left[ \xi_1 | (X,Y)=(x,y)\right] - \log f(x,y) \,\Big| \,\notag\\
&\leq& \delta_N + \frac{k(3\log N + |\log f(x,y)|)}{N^k} + \log N - \psi(N) + \frac{2C}{k+1} + 2\sqrt{C'(\frac{1}{N} + \frac{1}{k})}
\end{eqnarray}
for every $(x,y) \in \Omega_3 \setminus E$. By integration over $\Omega_3 \setminus E$, we have
\begin{eqnarray}
&&\int_{\Omega_3 \setminus E} \Big|\, \E \left[ \xi_1 | (X,Y)=(x,y)\right] - \log f(x,y) \,\Big| \,dP_{XY} \,\notag\\
&\leq& \int_{\Omega_3 \setminus E} \Big(\, \delta_N + \frac{k(3\log N + |\log f(x,y)|)}{N^k} + \log N - \psi(N)\,\notag\\
&&+\, \frac{2C}{k+1} + 2\sqrt{C'(\frac{1}{N} + \frac{1}{k})} \,\Big) dP_{XY}
\,\notag\\
&\leq& \delta_N + \frac{k(3\log N + \int_{\mathcal{X} \times \mathcal{Y}} |\log f(x,y)| dP_{XY})}{N^k} + \log N - \psi(N) \,\notag\\
&&+\, \frac{2C}{k+1} + 2\sqrt{C'(\frac{1}{N} + \frac{1}{k})} \;.
\end{eqnarray}
By Assumption 1, $k$ increases as $N \to \infty$. By Assumption 3, $\int_{\mathcal{X} \times \mathcal{Y}} |\log f(x,y)| dP_{XY} < +\infty$. Therefore, this quantity vanishes as $N \to \infty$. Combining with the case that $(x,y) \in E$, we have
\begin{eqnarray}
\lim_{N \to \infty} \int_{\Omega_3} \Big|\, \E \left[ \xi_1 | (X,Y)=(x,y)\right] - \log f(x,y) \,\Big| \,dP_{XY} = 0
\end{eqnarray}

\subsection{Proof of Lemma~\ref{lem:rnd}}
The proof of this lemma utilizes the Lebesgue-Besicovitch differentiation theorem~\cite[Theorem 1.32]{evans2018measure}, stated below
\begin{thm}[Lebesgue-Besicovitch Differentiation Theorem]
Let $\mu$ be a Radon measure on $\mathbb{R}^n$. For $f \in L_{loc}^1(\mu)$,
\begin{eqnarray}
    \lim_{r \to 0} \frac{1}{\mu(\bar{B}_r(x))} \int_{\bar{B}_r(x)} f d\mu = f(x),
\end{eqnarray}
for $\mu$-a.e. $x$.
\end{thm}

For our lemma, let $f = \frac{dP_{XY}}{dP_XP_Y}$ and $\mu = P_X P_Y$. Since $\mu$ is a probability measure, it is a Radon measure of Euclidean space. Also, since $\int_{\mathcal{X} \times \mathcal{Y}} |f| d \mu = 1$, so $f$ is globally integrable, hence locally integrable with respect to $\mu$. So the conditions of Lebesgue-Besicovitch differentiation theorem are satisfied, so
\begin{eqnarray}
f(x,y) &=& \frac{dP_{XY}}{dP_XP_Y}(x,y) \,\notag\\
&=& \lim_{r \to 0} \frac{1}{P_XP_Y(\bar{B}_r(x,y))} \int_{\bar{B}_r(x,y)} \frac{dP_{XY}}{dP_X P_Y} dP_XP_Y \,\notag\\
&=& \lim_{r \to 0} \frac{P_{XY}(\bar{B}_r(x,y))}{P_XP_Y(\bar{B}_r(x,y))} = \lim_{r \to 0}\frac{P_{XY}(x,y,r)}{P_X(x,r)P_Y(y,r)}
\end{eqnarray}

\subsection{Proof of Lemma~\ref{lem:lemma4}}
First, we upperbound $\mathbb{E} \left[ \log(X) \middle| X \geq k \right] - \log(Np)$. We can see that:

\begin{eqnarray}
& & \mathbb{E} \left[ X+m \middle| X \geq k \right] \\
& = & \frac{1}{ \mathbb{P} \left( X \geq k \right)} \sum_{i=k}^N (i+m) \left( \begin{array}{c} N \\ i \end{array}  \right) p^i (1-p)^{N-i} \\
& \leq & \frac{1}{1 - \exp\left( -2 \frac{(Np-k)^2}{N} \right)} \sum_{i=k}^N (i+m) \left( \begin{array}{c} N \\ i \end{array}  \right) p^i (1-p)^{N-i} \\
& \leq & \frac{1}{1 - \exp\left( -2 \frac{(Np-k)^2}{N} \right)} \sum_{i=1}^N (i+m) \left( \begin{array}{c} N \\ i \end{array}  \right) p^i (1-p)^{N-i} \\
& = & \frac{1}{1 - \exp\left( -2 \frac{(Np-k)^2}{N} \right)} \left( \mathbb{E}\left[ X \right] + m \right) = \frac{Np+m}{1 - \exp\left( -2 \frac{(Np-k)^2}{N} \right)}
\end{eqnarray}

In which we used the Hoeffding's inequality. Since $\mathbb{E} \left[ \log(X+m) \middle| X \geq k \right] \leq \log \left( \mathbb{E} \left[ X+m \middle| X \geq k \right] \right)$, thus:
\begin{eqnarray}
\mathbb{E} \left[ \log(X) \middle| X \geq k \right] - \log(Np) \leq \log \left( \frac{1+\frac{m}{Np}}{1 - \exp\left( -2 \frac{(Np-k)^2}{N} \right)} \right)
\end{eqnarray}

Second, to give an upper bound over $\log(Np) - \mathbb{E} \left[ \log(X+m) \middle| X \geq k \right]$, we first notice that:
\begin{equation} \log(Np) - \mathbb{E} \left[ \log(X+m) \middle| X \geq k \right] \leq \log(Np) - \mathbb{E} \left[ \log(X) \middle| X \geq k \right] \end{equation}

Then we upperbound $\log(Np) - \mathbb{E} \left[ \log(X) \middle| X \geq k \right]$ by applying Taylor's theorem around $x_0=Np$, where there exists $\zeta$ between $x$ and $x_0$ such that:
\begin{equation} \log(x) = \log(Np) + \frac{x-Np}{Np} - \frac{(x-Np)^2}{2\zeta^2} \end{equation}
since $\zeta \geq \min \left\{ x, x_0 \right\} = \min \left\{ x, Np \right\}$, we have:

\begin{eqnarray}
& & -\log(x) + \log(Np) + \frac{x-Np}{Np} = \frac{(x-Np)^2}{2 \zeta^2} \nonumber \\
& \leq & \max \left\{ \frac{(x-Np)^2}{2x^2}, \frac{(x-NP)^2}{2(Np)^2} \right\} \leq \frac{(x-Np)^2}{2x^2} + \frac{(x-Np)^2}{2(Np)^2 }
\end{eqnarray}

Now taking the conditional expectations from both sides, we have:

\begin{eqnarray}
 & & -\mathbb{E} \left[ \log(X) \middle| X \geq k \right] + \log(Np) +\frac{  \mathbb{E} \left[ X \middle| X \geq k \right] -Np}{Np} \nonumber \\
 & \leq & \mathbb{E} \left[ \frac{(X-Np)^2}{2X^2}  \middle| X \geq k \right] + \frac{ \mathbb{E} \left[ (X-Np)^2  \middle| X \geq k \right]}{2(Np)^2 }
\end{eqnarray}

First, we notice that $ \mathbb{E} \left[ X \middle| X \geq k \right] \geq  \mathbb{E} \left[ X \right] = Np$.

Second, $\mathbb{E} \left[ (X-Np)^2  \middle| X \geq k \right] \leq \frac{1}{1 - \exp\left( -2 \frac{(Np-k)^2}{N} \right)} \text{Var} \left[ X \right] = \frac{Np(1-p)}{1 - \exp\left( -2 \frac{(Np-k)^2}{N} \right)}$.

Thus we can write:
\begin{eqnarray}
-\mathbb{E} \left[ \log(X) \middle| X \geq k \right] + \log(Np) \leq \frac{Np(1-p)}{1 - \exp\left( -2 \frac{(Np-k)^2}{N} \right)}\frac{1}{2(Np)^2} + \mathbb{E} \left[ \frac{(X-Np)^2}{2X^2}  \middle| X \geq k \right] \label{eq:lemma_b3_modified}
\end{eqnarray}

To deal with the term $\mathbb{E} \left[ \frac{(X-Np)^2}{2X^2}  \middle| X \geq k \right]$, we have:

\begin{eqnarray}
\mathbb{E} \left[ \frac{(X-Np)^2}{2X^2}  \middle| X \geq k \right] &\leq& \frac{1}{1 - \exp\left( -2 \frac{(Np-k)^2}{N} \right)}
\sum_{i=k}^N \frac{(i-Np)^2}{2 i^2} \left( \begin{array}{c} N \\ i \end{array} \right) p^i (1-p)^{N-i} \\
&\leq& \frac{1}{1 - \exp\left( -2 \frac{(Np-k)^2}{N} \right)}
\sum_{i=k}^N \frac{(i-Np)^2}{(i+1)(i+2)} \left( \begin{array}{c} N \\ i \end{array} \right) p^i (1-p)^{N-i} \\
&=& \frac{1}{1 - \exp\left( -2 \frac{(Np-k)^2}{N} \right)}
\sum_{i=k}^N \frac{(i-Np)^2}{(N+1)(N+2)p^2} \left( \begin{array}{c} N+2 \\ i+2 \end{array} \right) p^{2+i} (1-p)^{N-i} \\
&\leq& \frac{1}{1 - \exp\left( -2 \frac{(Np-k)^2}{N} \right)}\frac{1}{(N+1)(N+2)p^2}
\mathbb{E}_{Y \sim \text{Bino}(N+2,p)} \left[ (Y-Np)^2 \right] \\
& = & \frac{1}{1 - \exp\left( -2 \frac{(Np-k)^2}{N} \right)}\frac{(N+2)p(1-p)+4p^2}{(N+1)(N+2)p^2} \\
& \leq & \frac{1}{1 - \exp\left( -2 \frac{(Np-k)^2}{N} \right)}\frac{(N+2)p}{(N+1)(N+2)p^2} \leq \frac{1}{1 - \exp\left( -2 \frac{(Np-k)^2}{N} \right)}\frac{1}{Np}
\end{eqnarray}

In which we used the fact that $ 2i^2 \geq (i+1)(i+2)$ for $i \geq 4$, and $(N+2)p \geq 4p$ for $N \geq 2$. Plugging it into Equation \ref{eq:lemma_b3_modified}, we have:

\begin{equation}
-\mathbb{E} \left[ \log(X) \middle| X \geq k \right] + \log(Np) \leq \frac{1}{1 - \exp\left( -2 \frac{(Np-k)^2}{N} \right)}\frac{3}{2Np}
\end{equation}

And the desired result is yielded.

\subsection{Proof of Lemma~\ref{lem:bino}}

Now we deal with the case that $\rho_{k,1} = r > 0$. Given that $(X_1, Y_1) = (x,y)$ and $\rho_{k,1} = r > 0$, we sort the samples $\{(X_i, Y_i)\}_{i=2}^N$ by their distance to $(x,y)$ defined as $d_i = \max\{\|X_i - x\|, \|Y_i - y\|\}$. To avoid the case that two samples have identical distance, we introduce a set of random variables $\{Z_i\}_{i=2}^N$ i.i.d. samples from ${\rm Unif}[0,1]$ and define a comparison operator $\prec$ as:
\begin{eqnarray}
i \prec j &\Longleftrightarrow& d_i < d_j {\rm \quad or \quad} \left\{ d_i = d_j {\rm \quad and \quad} Z_i < Z_j \right\}\;.
\end{eqnarray}
Since for any $i \neq j$, the probability that $Z_i = Z_j$ is zero, so we can have either $i \prec j$ or $i \succ j$ with probability 1. Now let $\{2, 3, \dots, N\} = S \cup \{j\} \cup T$ be a partition of the indices with $\left|S\right| = k-1$ and $\left|T\right| = N-k-1$.  Define an event $\mathcal{A}_{S, j, T}$ associated to the partition as:
\begin{eqnarray}
    \mathcal{A}_{S, j, T} = \big\{\, s \prec j, \forall s \in S,  \textrm{ and }t \succ j,\forall t \in T \,\big\}.
\end{eqnarray}
Since $(X_j,Y_j) - (x,y)$ are i.i.d.\ random variables each of the events $\mathcal{A}_{S, j, T}$ has identical probability. The number of all partitions is $\frac{(N-1)!}{(N-k-1)!(k-1)!}$ and thus  $\Pr\left(\,\mathcal{A}_{S, j, T}\,\right) = \frac{(N-k-1)!(k-1)!}{(N-1)!}$. So the cdf of $n_{x,1}$ is given by:
\begin{eqnarray}
&&\Pr\left(\,n_{x,1} \leq k+m\big|\rho_{k,1} = r, (X_1, Y_1) = (x,y)\,\right) \,\notag\\
&=&\sum_{S,j,T} \Pr\left(\,\mathcal{A}_{S, j, T}\,|\, \rho_{k,1} = r, (X_1, Y_1) = (x,y)\right) \Pr\left(\,n_{x,1} \leq k+m\big|\mathcal{A}_{S, j, T}, \rho_{k,1} = r, (X_1, Y_1) = (x,y)\,\right) \notag \\
&=&\frac{(N-k-1)!(k-1)!}{(N-1)!} \sum_{S,j,T} \Pr\left(\,n_{x,1} \leq k+m\big|\mathcal{A}_{S, j, T}, \rho_{k,1} = r, (X_1, Y_1) = (x,y)\,\right)
\end{eqnarray}

Now condition on event $\mathcal{A}_{S,j,T}$ and $\rho_{k,1} = r$, namely $(X_j, Y_j)$ is the $k$-nearest neighbor with distance $r$, $S$ is the set of samples with distance smaller than (or equal to) $r$ and $T$ is the set of samples with distance greater than (or equal to) $r$. Recall that $n_{x,1}$ is the number of samples with $\|X_j - x\| \leq r$. For any index $s \in S \cup \{j\}$, $\|X_j - x\| \leq r$ are satisfied. Therefore, $n_{x,1} \leq k+m$ means that there are no more than $m$ samples in $T$ with $\mathcal{X}$-distance smaller than $r$. Let $U_l = \mathbb{I}\{\|X_l - x\| \leq r \,\big|\, d_l \geq r\}.$ Therefore,
\begin{eqnarray}
    &&\Pr\left(\,n_{x,1} \leq k+m\big|\mathcal{A}_{S, j, T}, \rho_{k,1} = r, (X_1, Y_1) = (x,y)\,\right) \notag \\
    &=& \Pr\Big(\,\sum_{l \in T} \mathbb{I}\{\|X_l - x\| \leq r \} \leq m \,\big|\, d_s \leq r, \forall s \in S, d_j = r, d_t \geq r,\forall t \in T\,\Big) \notag \\
    &=& \Pr\left(\,\sum_{l \in T} \mathbb{I}\{\|X_l - x\| \leq r \} \leq m \,\big|\, d_l \geq r,\forall l \in T\,\right) = \Pr \left(\, \sum_{l \in T} U_l \leq m\,\right),
\end{eqnarray}
where $U_l$ follows bernoulli distribution with $\Pr\{U_l = 1\} = Pr\{\|X_l - x\| \leq r | d_l \geq r\}$. We can drop the conditioning of $(X_s,Y_s)$'s for $s \not\in T$ since $(X_s, Y_s)$ and $(X_t, Y_t)$ are independent. Therefore, given that $d_l \geq r$ for all $l \in T$, the variables $\mathbb{I}\{\|X_l - x\| \leq r \}$ are i.i.d. and have the same distribution as $U_l$. We conclude:
\begin{eqnarray}
    &&\Pr\left(\,n_{x,1} \leq k+m\big|\rho_{k,1} = r, (X_1, Y_1) = (x,y)\,\right) \,\notag\\
    &=& \frac{(N-k-1)!(k-1)!}{(N-1)!} \sum_{S,j,T} \Pr\left(\,n_{x,1} \leq k+m\big|\mathcal{A}_{S, j, T}, \rho_{i,xy} = r, (X_1, Y_1) = (x,y)\,\right) \notag \\
    &=& \frac{(N-k-1)!(k-1)!}{(N-1)!} \sum_{S,j,T}  \Pr \left(\, \sum_{l \in T} U_l \leq m\,\right)
    =  \Pr \left(\, \sum_{l \in T} U_l \leq m\,\right).
\end{eqnarray}
Thus we have shown that $n_{x,1}-k$ has the same distribution as $\sum_{l \in T} U_l$, which is a Binomial random variable with parameter $|T| = N-k-1$ and $\Pr\{\|X_l - x\| \leq r\,|\, d_l \geq r\} = \frac{P_X(x,r) - P_{XY}(x,y,r)}{1-P_{XY}(x,y,r)}$. For $n_{y,1}$, we can follow the same proof and conclude that $n_{y,1}-k \sim {\rm Bino}(N-k-1, \frac{P_Y(x,r) - P_{XY}(x,y,r)}{1-P_{XY}(x,y,r)})$.

\subsection{Proof of Lemma~\ref{lem:elog}}
By Jensen's inequality, we know that $\E [\log X] \leq \log \E [X] = \log (Np+m)$. So it suffices to give an upper bound for $\log (Np+m) - \E [\log X]$. We consider two different cases.
\\

(i) $Np \geq m$. In this case, for any $x$, by applying Taylor's theorem around $x_0 = Np+m$, there exists $\zeta$ between $x$ and $x_0$ such that
\begin{eqnarray}
\log (x) = \log (Np+m) + \frac{x-Np-m}{Np+m} - \frac{(x-Np-m)^2}{2\zeta^2}
\end{eqnarray}
By noticing that $\zeta \geq \min\{x, x_0\} = \min\{x, Np+m\}$, we have
\begin{eqnarray}
&&- \log (x) + \log(Np+m) + \frac{x-Np-m}{Np+m} = \frac{(x-Np-m)^2}{2\zeta^2} \,\notag\\
&\leq& \max\{\frac{(x-Np-m)^2}{2x^2}, \frac{(x-Np-m)^2}{2(Np+m)^2}\ \} \leq \frac{(x-Np-m)^2}{2x^2} + \frac{(x-Np-m)^2}{2(Np+m)^2}.
\end{eqnarray}
Now let $X-m$ be a ${\rm Bino}(N,p)$ random variable. By taking expectation on both sides, we have:
\begin{eqnarray}
&&\,- \E [\log X] + \log (Np+m) + \frac{\E[X] - Np - m}{Np+m} \,\notag\\
&\leq& \E \left[\, \frac{(X-Np-m)^2}{2X^2} \,\right] +    \frac{\E \left[\, (X-Np-m)^2\,\right]}{2(Np+m)^2} \label{eq:taylor_1}\;.
\end{eqnarray}
Since $\E[X] = Np+m$, $\E \left[\, (X-Np-m)^2\,\right] = \Var[X] = Np(1-p)$, and
\begin{eqnarray}
\E \left[\, \frac{(X-Np-m)^2}{2X^2} \,\right] &=& \sum_{j=0}^N \frac{(j-Np)^2}{2(j+m)^2} {N \choose j} p^j (1-p)^{N-j} \,\notag\\
&\leq& \sum_{j=0}^N \frac{(j-Np)^2}{(j+2)(j+1)} {N \choose j} p^j (1-p)^{N-j} \,\notag\\
&=& \sum_{j=0}^N \frac{(j-Np)^2}{(N+2)(N+1)p^2} {N+2 \choose j+2} p^{j+2} (1-p)^{N-j} \,\notag\\
&\leq& \frac{1}{(N+2)(N+1)p^2} \E_{Y \sim {\rm Bino}(N+2, p)} \left[\, (Y-Np)^2 \,\right] \,\notag\\
&=& \frac{(N+2)p(1-p)+4p^2}{(N+2)(N+1)p} \leq \frac{(N+2)p}{(N+2)(N+1)p} \leq \frac{1}{Np}
\end{eqnarray}
for $m \geq 1$ and $N \geq 4$. Plug these in~\eqref{eq:taylor_1}, we have
\begin{eqnarray}
&&\,- \E [\log X] + \log (Np+m) \leq \frac{1}{Np} + \frac{Np(1-p)}{2(Np+m)^2} \,\notag\\
&\leq& \frac{2}{Np+m} + \frac{1}{2(Np+m)} = \frac{5}{2(Np+m)}\;.
\end{eqnarray}
where $1/(2Np) \leq 1/(Np+m)$ comes from the fact that $Np \geq m$.
\\

(ii) $Np < m$. In this case, for any $x$, by applying Taylor's theorem around $x_0 = Np+m$, there exists $\zeta$ between $x$ and $x_0$ such that
\begin{eqnarray}
\log (x) = \log (Np+m) + \frac{x-Np-m}{Np+m} - \frac{(x-Np-m)^2}{2\zeta^2}
\end{eqnarray}
By noticing that $\zeta \geq \min\{x, x_0\} \geq m \geq (Np+m)/2$, we have:
\begin{eqnarray}
- \log (x) + \log (Np+m) + \frac{x-Np-m}{Np+m} \leq \frac{2(x-Np-m)^2}{(Np+m)^2} \;.
\end{eqnarray}
Similarly, by taking expectation on both sides, we have
\begin{eqnarray}
- \E [\log X] + \log (Np+m) + \frac{\E[X] - Np - m}{Np+m} \leq \frac{\E \left[\, 2(X-Np-m)^2\,\right]}{(Np+m)^2} \;.
\end{eqnarray}
By plugging in $\E[X] = Np+m$ and $\E \left[\, (X-Np-k)^2\,\right] = \Var[X] = Np(1-p)$, we obtain
\begin{eqnarray}
- \E [\log X] + \log (Np+m) \leq \frac{ 2Np(1-p) }{(Np+m)^2} \leq \frac{2(Np+m)}{(Np+m)^2} = \frac{2}{Np+m} \;.
\end{eqnarray}

Combining the two cases, we obtain the desired statement.

\section{Proof of Theorem~\ref{thm:var}}
\label{sec:var}
We use the Efron-Stein inequality to bound the variance of the estimator. For simplicity, let $\widehat{I}^{(N)}(Z)$ be the estimate based on original samples $\{Z_1, Z_2, \dots, Z_N\}$, where $Z_i = (X_i, Y_i)$. For the usage of Efron-Stein inequality, we consider another set of i.i.d. samples $\{Z'_1, Z'_2, \dots, Z'_n\}$ drawn from $P_{XY}$. Let $\widehat{I}^{(N)}(Z^{(j)})$ be the estimate based on $\{Z_1, \dots, Z_{j-1}, Z'_j, Z_{j+1}, \dots, Z_N\}$. Then Efron-Stein inequality states that
\begin{eqnarray}
\Var \left[\, \widehat{I}^{(N)}(Z) \,\right] &\leq& \frac{1}{2} \sum_{j=1}^N \E \left[\, \left(\, \widehat{I}^{(N)}(Z) - \widehat{I}^{(N)}(Z^{(j)})\,\right)^2 \,\right] \;. \label{eq:E_S}
\end{eqnarray}

Now we will give an upper bound for the difference $|\widehat{I}^{(N)}(Z) - \widehat{I}^{(N)}(Z^{(j)})|$ for given index $j$. First of all, let $\widehat{I}^{(N)}(Z_{\setminus j})$ be the estimate based on $\{Z_1, \dots, Z_{j-1}, Z_{j+1}, \dots, Z_N\}$, then by triangle inequality, we have:
\begin{eqnarray}
&&\sup_{Z_1, \dots, Z_N, Z'_j} \Big|\, \widehat{I}^{(N)}(Z) - \widehat{I}^{(N)}(Z^{(j)}) \,\Big| \,\notag\\
&\leq& \sup_{Z_1, \dots, Z_N, Z'_j} \left(\, \Big|\, \widehat{I}^{(N)}(Z) - \widehat{I}^{(N)}(Z_{\setminus j}) \,\Big| + \Big|\, \widehat{I}^{(N)}(Z_{\setminus j}) - \widehat{I}^{(N)}(Z^{(j)}) \,\Big| \,\right) \,\notag\\
&\leq& \sup_{Z_1, \dots, Z_N} \Big|\, \widehat{I}^{(N)}(Z) - \widehat{I}^{(N)}(Z_{\setminus j}) \,\Big| + \sup_{Z_1, \dots, Z_{j-1}, Z'_j, Z_{j+1}, \dots, Z_N} \Big|\, \widehat{I}^{(N)}(Z_{\setminus j}) - \widehat{I}^{(N)}(Z^{(j)}) \,\Big| \,\notag\\
&=& 2 \sup_{Z_1, \dots, Z_N} \Big|\, \widehat{I}^{(N)}(Z) - \widehat{I}^{(N)}(Z_{\setminus j}) \,\Big|
\end{eqnarray}
where the last equality comes from the fact that $\{Z_1, \dots, Z_{j-1}, Z'_j, Z_{j+1}, \dots, Z_N\}$ has the same joint distribution as $\{Z_1, \dots, Z_N\}$. Now recall that \begin{eqnarray}
\widehat{I}^{(N)}(Z) = \frac{1}{N} \sum_{i=1}^N \xi_i(Z) = \frac{1}{N} \sum_{i=1}^N \left(\, \psi(\tilde{k}_i) + \log N - \log (n_{x,i}+1) - \log (n_{y,i}+1) \,\right) \;,
\end{eqnarray}
Therefore, we have
\begin{eqnarray}
\sup_{Z_1, \dots, Z_N, Z'_j} \Big|\, \widehat{I}^{(N)}(Z) - \widehat{I}^{(N)}(Z^{(j)}) \,\Big| \leq \frac{2}{N} \sup_{Z_1, \dots, Z_N} \sum_{i=1}^N  \Big|\, \xi_i(Z) - \xi_i(Z_{\setminus j}) \,\Big| \;. \label{eq:diff}
\end{eqnarray}
Now we need to upper-bound the difference $|\, \xi_i(Z) - \xi_i(Z_{\setminus j}) \,|$ created by eliminating sample $Z_j$ for different $i$ 's. There are three cases of $i$'s as follows,
\begin{itemize}
    \item {\bf Case I.} $i = j$. Since the upper bounds $|\xi_i(Z)| \leq 2 \log N$ and $|\xi_i(Z_{\setminus j})| \leq 2 \log (N-1)$ always holds, so $|\, \xi_i(Z) - \xi_i(Z_{\setminus j}) \,| \leq 4 \log N$. The number of $i$'s in this case is only 1. So $\sum_{\textrm{Case I}} |\, \xi_i(Z) - \xi_i(Z_{\setminus j}) \,| \leq 4 \log N$.
    \item {\bf Case II.} $\rho_{i,xy} = 0$. In this case, recall that $\tilde{k}_i = \Big|\, \{i' \neq i: Z_i = Z_{i'}\} \,\Big|$, $n_{x,i} = \Big|\, \{i' \neq i: X_i = X_{i'}\} \,\Big|$ and $n_{y,i} = \Big|\, \{i' \neq i: Y_i = Y_{i'}\} \,\Big|$. There are 4 sub-cases in this case.
        \begin{itemize}
            \item {\bf Case II.1.} $Z_i = Z_j$. By eliminating $Z_j$, $\tilde{k}_i$, $n_{x,i}$, $n_{y,i}$ will all decrease by 1. Therefore,
                \begin{eqnarray}
                &&|\, \xi_i(Z) - \xi_i(Z_{\setminus j}) \,| \,\notag\\
                &=& |\, \left(\, \psi(\tilde{k}_i) + \log N - \log(n_{x,i}+1) - \log(n_{y,i}+1)\,\right) \,\notag\\
                &&-\, \left(\, \psi(\tilde{k}_i-1) + \log (N-1) - \log(n_{x,i}) - \log(n_{y,i})\,\right) \,| \,\notag\\
                &\leq& |\psi(\tilde{k}_i) - \psi(\tilde{k}_i-1)| + |\log N - \log(N-1)| \,\notag\\
                &&+\, |\log(n_{x,i}+1) - \log(n_{x,i})| + |\log(n_{y,i}+1) - \log(n_{y,i})| \,\notag\\
                &\leq& \frac{1}{\tilde{k}_i - 1} + \frac{1}{N-1} + \frac{1}{n_{x,i}} + \frac{1}{n_{y,i}} \leq \frac{4}{\tilde{k}_i - 1} = \frac{4}{\tilde{k}_j - 1}\;.
                \end{eqnarray}
                The number of $i$'s in this case is the number if $i$'s such that $Z_i = Z_j$, which is just $\tilde{k}_j$. Therefore, $\sum_{\textrm{Case II.1}} |\, \xi_i(Z) - \xi_i(Z_{\setminus j}) \,| \leq 4\tilde{k}_j/(\tilde{k}_j - 1) \leq 8$, for $\tilde{k}_j \geq k \geq 2$.
            \item {\bf Case II.2.} $X_i = X_j$ but $Y_i \neq Y_j$. By eliminating $Z_j$, $\tilde{k}_i$ and $n_{y,i}$ won't change but $n_{x,i}$ will decrease by 1. Therefore,
                \begin{eqnarray}
                |\, \xi_i(Z) - \xi_i(Z_{\setminus j}) \,| &\leq& |\log N - \log (N-1)| + |\log(n_{x,i}+1) - \log(n_{x,i})| \,\notag\\
                &\leq& \frac{1}{N-1} + \frac{1}{n_{x,i}} \leq \frac{2}{n_{x,i}} = \frac{2}{n_{x,j}}
                \end{eqnarray}
                The number of $i$'s in this case is the number if $i$'s such that $X_i = X_j$ but $Y_i \neq Y_j$, which is less than $n_{x,j}$. Therefore, $\sum_{\textrm{Case II.2}} |\, \xi_i(Z) - \xi_i(Z_{\setminus j}) \,| \leq 2n_{x,j}/n_{x,j} \leq 2$.
            \item {\bf Case II.3.} $Y_i = Y_j$ but $X_i \neq X_j$. By eliminating $Z_j$, $\tilde{k}_i$ and $n_{x,i}$ won't change but $n_{y,i}$ will decrease by 1. Similarly as Case II.2, we have $\sum_{\textrm{Case II.3}} |\, \xi_i(Z) - \xi_i(Z_{\setminus j}) \,| \leq 2$.
            \item {\bf Case II.4.} $X_i \neq X_j$ and $Y_i \neq Y_j$. In this case, none of $\tilde{k}_i$, $n_{x,i}$, or $n_{y,i}$ will change. So $|\, \xi_i(Z) - \xi_i(Z_{\setminus j}) \,| = \log N - \log(N-1) \leq 1/(N-1)$. The number of $i$'s in this case is simply less than $N-1$. Therefore, $\sum_{\textrm{Case II.4}} |\, \xi_i(Z) - \xi_i(Z_{\setminus j}) \,| \leq 1$.
        \end{itemize}
        Combining the four sub-cases, we conclude that $\sum_{\textrm{Case II}} |\, \xi_i(Z) - \xi_i(Z_{\setminus j}) \,| \leq 13$.
    \item {\bf Case III.} $\rho_{i,xy} > 0$. In this case, recall that $\tilde{k}_i$ always equals to $k$, $n_{x,i} = \Big|\, \{i' \neq i: \|X_i - X_{i'}\| \leq \rho_{i,xy}\} \,\Big|$ and $n_{y,i} = \Big|\, \{i' \neq i: \|Y_i - Y_{i'}\| \leq \rho_{i,xy}\} \,\Big|$. Similar to Case II, there are 4 sub-cases.
        \begin{itemize}
            \item {\bf Case III.1.} $Z_j$ is in the $k$-nearest neighbors of $Z_i$. In this case, we don't know how $n_{x,i}$ and $n_{y,i}$ will change by eliminating $Z_j$, so we just use the loosest bound $|\, \xi_i(Z) - \xi_i(Z_{\setminus j}) \,| \leq 4 \log N$. However, the number of $i$'s in this case is upper bounded by the following lemma.
                \begin{lemma}
                \label{lem:knn}
                Let $Z, Z_1, Z_2, \dots, Z_N$ be vectors of $\mathbb{R}^d$ and $\mathcal{Z}_i$ be the set $\{Z_1, \dots, Z_{i-1},Z,Z_{i+1}, \dots, Z_N\}$. Then
                \begin{eqnarray}
                \sum_{i=1}^N \mathbb{I}\{Z \textrm{ is in the } k \textrm{-nearest neighbors of } Z_i \textrm{ in } \mathcal{Z}_i\} \leq k \gamma_d \;,
                \end{eqnarray}
                (distance ties are broken by comparing indices). Here $\gamma_d$ is the minimum number of cones with angle smaller than $\pi/6$ needed to cover $\mathbb{R}^d$.  Moreover, if we allow $k$ to be different for difference $i$, we have
                \begin{eqnarray}
                \sum_{i=1}^N \frac{1}{k_i} \mathbb{I}\{Z \textrm{ is in the } k_i \textrm{-nearest neighbors of } Z_i \textrm{ in } \mathcal{Z}_i\} \leq \gamma_d (\log N + 1)\;.
                \end{eqnarray}
                \end{lemma}
                By the first inequality in Lemma~\ref{lem:knn}, the number of $i$'s in this case is upper bounded by $k \gamma_d$. Therefore, $\sum_{\textrm{Case III.1}} |\, \xi_i(Z) - \xi_i(Z_{\setminus j}) \,| \leq 4 k \gamma_{d_x+d_y} \log N$.
            \item {\bf Case III.2.} $Z_j$ is not in the $k$-nearest neighbors of $Z_i$, but $\|X_j - X_i\| \leq \rho_{i,xy}$, i.e., $X_j$ is in the $n_{x,i}$-nearest neighbors of $X_i$. In this case, $n_{x,i}$ will decrease by 1 and $n_{y,i}$ remains the same. So
                \begin{eqnarray}
                |\, \xi_i(Z) - \xi_i(Z_{\setminus j}) \,| &\leq& |\log N - \log (N-1)| + |\log(n_{x,i}+1) - \log(n_{x,i})| \,\notag\\
                &\leq& \frac{1}{N-1} + \frac{1}{n_{x,i}} \leq \frac{2}{n_{x,i}}
                \end{eqnarray}
                We don't have an upper bound for the number of $i$'s in this case, but from the second inequality in Lemma~\ref{lem:knn}, we have the following upper bound, where $\mathcal{X}_{i,j} = \{X_1, \dots, X_{i-1},X_j,X_{i+1}, \dots, X_N\}$:
                \begin{eqnarray}
                &&\sum_{\textrm{Case III.2}} |\, \xi_i(Z) - \xi_i(Z_{\setminus j}) \,| \,\notag\\
                &\leq& \sum_{i=1}^N \frac{2}{n_{x,i}} \mathbb{I}\{X_j \textrm{ is in the } n_{x,i} \textrm{-nearest neighbors of } X_i \textrm{ in } \mathcal{X}_{i,j}\} \,\notag\\
                &\leq& 2 \gamma_{d_x} (\log N + 1) \leq 2 \gamma_{d_x+d_y} (\log N + 1)\;.
                \end{eqnarray}
            \item {\bf Case III.3.} $Z_j$ is not in the $k$-nearest neighbors of $Z_i$, but $\|Y_j - Y_i\| \leq \rho_{i,xy}$, i.e., $Y_j$ is in the $n_{y,i}$-nearest neighbors of $Y_i$. In this case, $n_{y,i}$ will decrease by 1 and $n_{x,i}$ remains the same. Follow the same analysis in Case III.2, we have $\sum_{\textrm{Case III.2}} |\, \xi_i(Z) - \xi_i(Z_{\setminus j}) \,| \leq 2 \gamma_{d_x+d_y} (\log N + 1)$ as well.
            \item {\bf Case III.4.} $Z_j$ is not in the $k$-nearest neighbors of $Z_i$, and $\|X_j - X_i\| > \rho_{i,xy}$, $\|Y_j - Y_i\| > \rho_{i,xy}$. In this case, neither $n_{x,i}$ nor $n_{y,i}$ will change. Similar to Case II.4, $\sum_{\textrm{Case III.4}} |\, \xi_i(Z) - \xi_i(Z_{\setminus j}) \,| \leq 1$.
        \end{itemize}
        Combining the four sub-cases, we conclude that $\sum_{\textrm{Case III}} |\, \xi_i(Z) - \xi_i(Z_{\setminus j}) \,| \leq (4k+4)\gamma_{d_x+d_y} \log N + 4\gamma_{d_x+d_y} + 1$.
\end{itemize}
Combining the three cases, we have:
\begin{eqnarray}
\sum_{i=1}^N  \Big|\, \xi_i(Z) - \xi_i(Z_{\setminus j}) \,\Big| & \leq& 4 \log N + 13 + (4k+4) \gamma_{d_x+d_y} \log N + 4 \gamma_{d_x+d_y} + 1 \,\notag\\
&\leq& 30 \gamma_{d_x+d_y} k \log N
\end{eqnarray}
for $k \geq 1$, $\log N \geq 1$ and all $\{Z_1, \dots, Z_N\}$. Plug it into~\eqref{eq:diff}, we obtain,
\begin{eqnarray}
\sup_{Z_1, \dots, Z_N, Z'_j} \Big|\, \widehat{I}^{(N)}(Z) - \widehat{I}^{(N)}(Z^{(j)}) \,\Big| \leq \frac{60 \gamma_{d_x+d_y} k \log N}{N} \;.
\end{eqnarray}
Plug it into Efron-Stein inequality~\eqref{eq:E_S}, we obtain:
\begin{eqnarray}
\Var \left[\, \widehat{I}^{(N)}(Z) \,\right] &\leq& \frac{1}{2} \sum_{j=1}^N \E \left[\, \left(\, \widehat{I}^{(N)}(Z) - \widehat{I}^{(N)}(Z^{(j)})\,\right)^2 \,\right] \,\notag\\
&\leq& \frac{1}{2} \sum_{j=1}^N \sup_{Z_1, \dots, Z_n, Z'_j} \left(\, \widehat{I}^{(N)}(Z) - \widehat{I}^{(N)}(Z^{(j)}) \,\right)^2 \,\notag\\
&\leq& \frac{1}{2} \sum_{j=1}^N (\frac{60 \gamma_{d_x+d_y} k \log N}{N})^2 = \frac{1800 \gamma_{d_x+d_y}^2 (k \log N)^2}{N} \;.
\end{eqnarray}
Since $1800\gamma^2_{d_x+d_y}$ is a constant independent of $N$, and $(k_N \log N)^2 / N \to 0$ as $N \to \infty$ by Assumption 6, we have $\lim_{N \to \infty} \Var \left[\, \widehat{I}^{(N)}(Z) \,\right] = 0$.

\subsection{Proof of Lemma~\ref{lem:knn}}
For the first part of the lemma, we refer to Lemma 20.6 in~\cite{biau2015lectures}.

The second part of the lemma is a consequence of the first part. We reorder the indices $i$'s by $k_i$ and rewrite the summation as follows,
\begin{eqnarray}
&&\sum_{i=1}^N \frac{1}{k_i} \mathbb{I}\{Z \textrm{ is in the } k_i \textrm{-nearest neighbors of } Z_i \textrm{ in } \mathcal{Z}_i\} \,\notag\\
&=& \sum_{k=1}^N \frac{1}{k} \sum_{i=1}^N \mathbb{I}\{k_i = k\} \mathbb{I}\{Z \textrm{ is in the } k \textrm{-nearest neighbors of } Z_i \textrm{ in } \mathcal{Z}_i\} \,\notag\\
&=& \sum_{k=1}^N \frac{1}{k} \sum_{i=1}^N \mathbb{I}\{k_i = k \textrm{ and } Z \textrm{ is in the } k \textrm{-nearest neighbors of } Z_i \textrm{ in } \mathcal{Z}_i\}
\end{eqnarray}
Notice that we take the summation over $k=1$ to $N$ since each $k_i$ can not be more than $N$. Denote $S_k = \sum_{i=1}^N \mathbb{I}\{k_i = k \textrm{ and } Z \textrm{ is in the } k \textrm{-nearest neighbors of } Z_i \textrm{ in } \{Z_1, \dots, Z_{i-1}, Z, Z_{i+1}, \dots, Z_N\}\}$ for simplicity. Then we need to prove that $\sum_{k=1}^N (S_k/k) \leq \gamma_d \log N$. By the first part of this lemma, we obtain,
\begin{eqnarray}
\sum_{\ell=1}^k S_{\ell} &=& \sum_{\ell=1}^k \sum_{i=1}^N \mathbb{I}\{k_i = \ell \textrm{ and } Z \textrm{ is in the } \ell \textrm{-nearest neighbors of } Z_i \textrm{ in } \mathcal{Z}_i\} \,\notag\\
&=& \sum_{i=1}^N \sum_{\ell=1}^k \mathbb{I}\{k_i = \ell \textrm{ and } Z \textrm{ is in the } \ell \textrm{-nearest neighbors of } Z_i \textrm{ in } \mathcal{Z}_i\} \,\notag\\
&\leq& \sum_{i=1}^N \mathbb{I}\{k_i \leq k \textrm{ and } Z \textrm{ is in the } k \textrm{-nearest neighbors of } Z_i \textrm{ in } \mathcal{Z}_i\} \,\notag\\
&\leq& k \gamma_d \;.
\end{eqnarray}
Therefore, we obtain
\begin{eqnarray}
&& \sum_{k=1}^N \frac{S_k}{k} = \sum_{k=1}^{N-1} \frac{1}{k(k+1)} \left(\, \sum_{\ell=1}^k S_{\ell} \,\right) + \frac{1}{N} \sum_{\ell=1}^N S_{\ell} \,\notag\\
&\leq& \sum_{k=1}^{N-1} \frac{k\gamma_d}{k(k+1)} + \frac{N \gamma_d}{N} = \sum_{k=1}^{N} \frac{\gamma_d}{k} < \gamma_d (\log N + 1) \;,
\end{eqnarray}
which completes the proof.


{
\bibliographystyle{plain}
\bibliography{mixed}
}

\end{document}